\documentclass[a4paper, 12pt]{article} 

\usepackage{fullpage}
\usepackage[utf8]{inputenc}
\usepackage[T1]{fontenc}
\usepackage[UKenglish]{babel}
\usepackage{lmodern}

\usepackage{nicefrac}
\usepackage{mathtools}
\usepackage{amsfonts, dsfont}
\usepackage{amsthm, stmaryrd}
\usepackage{verbatim, hyperref}
\hypersetup{pdfborder=0 0 0}
\usepackage{ulem}
\usepackage[dvipsnames]{xcolor}
\usepackage{subcaption}
\usepackage{soul}
\usepackage{natbib}

\usepackage{graphics}
\usepackage{graphicx}
\usepackage{svg}
\usepackage[font=footnotesize]{caption}
\graphicspath{ {./pics5/} }
\usepackage{wrapfig, lipsum}

\usepackage{floatrow}
\usepackage{subcaption}

\usepackage[plain]{algorithm}
\usepackage{algorithmic}

\newtheorem{prop}{Proposition}[section]

\newtheorem{lem}[prop]{Lemma}
\newtheorem{thm}[prop]{Theorem}

\newtheorem{ass}[prop]{Assumption}

\theoremstyle{remark}
\newtheorem{rem}[prop]{Remark}
\newtheorem{example}[prop]{Example}

\newcommand{\sen}{\mathtt{na}}
\newcommand{\nsen}{\mathtt{a}}

\newcommand{\overhang}{\omega}

\newcommand{\kb}{\text{kb}}
 
\title{Stochastic branching models for the telomeres dynamics in a model including telomerase activity} 
\author{A. Benetos$^{1,2}$, C. Fritsch$^{3}$, E. Horton$^{4}$, L. Lenotre$^{5}$,\\ S. Toupance$^{1}$ and D. Villemonais$^{6,7}$}

\begin{document}

 \footnotetext[1]{Université de Lorraine, Inserm, DCAC, F-54000 Nancy, France}
 \footnotetext[2]{Université de Lorraine, CHRU-Nancy, Service de Gériatrie, F-54000 Nancy, France}
 \footnotetext[3]{Université de Lorraine, CNRS, Inria, IECL, F-54000 Nancy, France}
\footnotetext[4]{Department of Statistics, University of Warwick, Coventry, CV4 7AL, UK}
\footnotetext[5]{Université de Haute-Alsace, Faculté des Sciences et Techniques, F-68200 Mulhouse, France}
\footnotetext[6]{Université de Strasbourg, IRMA, F-67000 Strasbourg, France}
\footnotetext[7]{Institut universitaire de France (IUF)}

\maketitle

\begin{abstract}
Telomeres are repetitive sequences of nucleotides at the end of chromosomes, whose evolution over time is intrinsically related to biological ageing. In most cells, with each cell division, telomeres shorten due to the so-called end replication problem,
which can lead to replicative senescence and a variety of age-related diseases. 
On the other hand, in certain cells, the presence of the enzyme telomerase can lead to the lengthening of telomeres, which may delay or prevent the onset of such diseases but can also increase the risk of cancer.

In this article, we propose a stochastic representation of this biological model, which takes into account multiple chromosomes per cell, the effect of telomerase, different cell types and the dependence of the distribution of telomere length on the dynamics of the process. We study theoretical properties of this model, including its long-term behaviour. In addition, we investigate numerically the impact of the model parameters on biologically relevant quantities, such as the Hayflick limit and the Malthusian parameter of the population of cells.
\end{abstract}



\section{Introduction}\label{sec:intro}

 Telomeres are repetitive sequences of nucleotides located at the ends of linear chromosomes that act as protective caps to ensure the genomic stability. 
 As cells divide, telomeres gradually shorten, ultimately reaching a critical length triggering cellular senescence or apoptosis~\citep{OLOVNIKOV1996443,EntringerPunderEtAl2018,xu2019many}, processes implicated in a number of age-related diseases such as cardiovascular, metabolic and neurodegenerative diseases \citep{Dmello2015, Haycock2017, Toupance_Benetos_2019}.
In addition, research has shown that factors such as lifestyle, genetics and oxidative stress can also 
impact telomere length dynamics \citep{StarkweatherAlhaeeriEtAl2014, Shammas2011, Louzon2019}. 

On the other hand, some cell types can express an enzyme called telomerase. This DNA polymerase has the ability to add telomeric repeats to the end of telomeres during DNA replication, which can compensate the effects of telomere shortening but also plays an important role in the development and evolution of cancer~\citep{jafri2016roles,blasco2005telomeres}.

In the past 30 years, mathematical modelling has played an important role in understanding the long-term behaviour of the distribution of telomere length in cells and the associated health implications.
One of the earliest (deterministic) models for the evolution of telomere lengths was given by \cite{Levy1992}, which mirrored the results seen in \textit{in vitro} experiments. Shortly after, \cite{Arino1995} and \cite{Olofsson1998} reframed this model in terms of a branching process to obtain both exact and asymptotic results for the behaviour of telomeres. Since then, there have been a wide range of stochastic models for telomere lengths, both with and without the presence of telomerase, for the purpose of understanding the role of the shortest telomere in senescence \citep{BourgeronXuEtAl2015, DaoDuc2013}, calculating the time until senescence \citep{Eugene2017}, parameter inference \citep{LeeKimmel2020} and calculating the stationary distribution \citep{LeeKimmel2020}. We also refer the reader to the articles of \cite{Portugal2008, WattisQiEtAl2020, Wattis2014} for computational analysis of both deterministic and stochastic models.

In this article, our first main contribution is to propose a mathematical model that describes the evolution in continuous time of a population of cells, whose dynamics depend explicitly on the lengths of their telomeres. This is in accordance with recent empirical evidence \citep{BourgeronXuEtAl2015, DaoDuc2013, XuDucEtAl2013}, which suggests a strong link between the length of the shortest telomere and the cell's behaviour. The main novelty in our approach is that we allow for different cell types, multiple chromosomes per cell {\it and} we also take into account the effect of telomerase.  Moreover, we model the dynamics of the {\it whole population} of cells, which requires one to keep track of the state of all telomeres in each individual cell within a growing cell population. To the best of our knowledge, this is the first mathematical model that takes all of these components into account simultaneously.

Among the models in the aforementioned literature, some of them describe the whole cell population, while others describe a single cell lineage.
While both approaches are relevant, we emphasize that the latter approach does not appropriately describe the distribution of the whole population when reproduction and/or senescence rates of a cell depend on its telomere lengths. This is due to an imbalance of mass compared to the whole population of cells. To counteract this issue, one needs to appropriately weight the single linear process, as in the \textit{many-to-one formula}. In the setting of telomere dynamics, this is given in Lemma~\ref{lem:manytoone}, which allows one to represent the first moment of the entire cell population by the average behaviour of a typical trajectory in the population, appropriately weighted. This result allows one to more easily analyse the behaviour of the branching process, and provides a useful tool for numerical simulations.

One of the main mathematical contributions of this paper, Theorem~\ref{thm:mainresult}, is related to the asymptotic behaviour of the first moment of the particle system. In the case that telomerase is present and compensates the attrition of telomeres, so that the population grows indefinitely, we show two things. Firstly, the distribution of an average cell and its telomere lengths converges with time, and secondly, the first moment of the population size grows asymptotically exponentially fast.  This analysis is related to the theory of quasi-stationary distributions (see e.g.~\cite{ColletMartinezEtAl2013, DoornPollett2013,MeleardVillemonais2012}). While the analysis presented in this article is fairly standard in the literature on quasi-stationary distributions, to the best of our knowledge, this is the first time these techniques have been applied in the context of telomere dynamics.

We also present numerical simulations of the model, both with and without telomerase. In the case where it has very little impact or is not present at all, the population eventually stabilizes due to the fact that cells can no longer divide when their telomeres are too short. In this case, the number of times the population doubles is finite; this effect is also called the \textit{Hayflick limit} \citep{HayflickMoorhead1961, ShayWright2000, trentesaux2010senescence}. Thus, we investigate the influence of the model parameters on the Hayflick limit. 
On the other hand, when telomerase is present, we study its influence on the asymptotic growth rate of the population and the stabilization of the distribution of an average cell and its telomere lengths, illustrating Theorem~\ref{thm:mainresult} in practice.

The rest of the paper is set out as follows. In Section~\ref{sec:model}, we present the specific model that we will work with throughout the paper, presenting first the biological mechanism for DNA replication at the level of chromosomes, followed by the stochastic process describing the population of cells. This model takes into account the telomere length distribution of chromosomes, as well as the cell's ability to divide or not. In addition, we allow for the possibility that the dynamics depend on the absence/presence of telomerase.
Section~\ref{sec:semigroup} is devoted to the mathematical analysis of the first moment of the mathematical model, leading to the {\it many-to-one formula} (see Lemma~\ref{lem:manytoone}). In Section~\ref{sec:QSD}, we establish and prove our main mathematical result, which characterises the long-term behaviour of our process (see Theorem~\ref{thm:mainresult}). In particular, we describe the growth rate of the average number of particles in the system and the corresponding stationary distribution. 
Section~\ref{sec:simulations} contains the numerical simulations, which illustrate the impact of the model parameters on 1) the Hayflick limit (the number of times the population doubles), 2) the asymptotic growth rate of the population, and 3) the asymptotic distribution of telomere lengths in the cell population.
Finally, in Section~\ref{sec:discussion}, we discuss our model, results and possible future work and extensions. In addition, we alert the reader to the table of notation in Appendix \ref{sec:notation}, which collects the main parameters used for the simulations in Section \ref{sec:simulations}.

\section{The model}
\label{sec:model}
\subsection{The biological model}\label{sec:bio}
In this section, we describe the biological DNA replication mechanism we adhere to throughout the article. 
For concreteness, we will describe the replication mechanism that occurs in human cells, however, the mathematical model we consider in the next section is robust enough to apply to other organisms. 
We will also only describe the relevant aspects of the DNA replication process, and we refer the reader to the articles of \cite{Grach2013, Aubert2014} and references therein for further details. 

Human cells contain 23 pairs of chromosomes, each composed of a DNA strand whose extremities are called telomeres. Each strand of the DNA double helix has a 3' and a 5' end. Due to the antiparallel structure of the double helix, the 3' end of one strand opposes the 5' end of the other strand. Moreover, the 3' end overhangs the opposing 5' end (see Figure~\ref{fig.chrom.before.div}). 

Starting from the centre of the chromosome (referred to as the origin of replication), the DNA helix is `unzipped' in the direction of the telomeres, creating two `Y' shapes, each called a replication fork (see Figure~\ref{fig.rep.fork}). 
The two separated strands will then act as templates for an enzyme (DNA polymerase) to make new complementary DNA strands, resulting in two copies of the original chromosome. 

The orientation of the DNA helix means that when it is unzipped, in one direction it unzips from the 3' to the 5' end, which will form the {\it leading} strand, while in the other direction the helix is unzipped in the 5' to 3' direction, and will form the {\it lagging} strand. Due to the way DNA polymerase replicates DNA and the antiparallel nature of a DNA helix, the mechanism is different for each of the two strands. Indeed, as the DNA polymerase synthesises the leading strand, DNA fragments are added continuously to form the new complementary (antiparallel) strand from the origin to the end of the 5' end of the template. Note that the new complementary strand will be shorter than the original one since DNA polymerase can only use the 3' to 5' strand as a template, resulting in the loss of the original overhang. In order to re-establish the overhang, an enzyme called nuclease removes the end of the template 5' end so that the new 3' is longer. On the other hand, when DNA polymerase synthesises the lagging strand, it does so in separated segments, called Okazaki fragments, using a `back-stitch' type method. Due to the discontinuous nature of the replication along the lagging strand, DNA cannot be replicated all the way to the end of the template strand, hence also resulting in an overhang of the 3' end. We refer the reader to Figures~\ref{fig.chrom.before.div}, \ref{fig.rep.fork} and~\ref{fig.repli.before.tel} for a diagram of these steps, as well as the articles of \cite{Grach2013} and~\cite{MurakiEtAl2012} for further details and references.

Thus, as a result of the mechanism described above, telomeres become shorter with each cell division. This is known as the \textit{end replication problem}. When the shortest telomere in a cell becomes too short, the cell is unable to divide any further (see  \cite{BourgeronXuEtAl2015, Hemann2001}), since another division would risk damaging the DNA. 
When this occurs, the cell becomes senescent, which we refer to as deactivation in the rest of the paper. Thus, if the shortest telomere in each cell falls below a certain threshold, the number of cells in the population ceases to increase. This leads to a concept called the Hayflick limit, which is defined as the number of times a population of cells can double (before cell division is no longer possible).

In certain cells, such as stem cells, the majority of cancer cells and some somatic cells, an enzyme called telomerase is present, which provides a mechanism by which telomeres can lengthen. 
In these cells, after DNA replication has occurred, telomerase has the ability to add new nucleotides to the 3' end. After a certain number of nucleotides have been added, the complementary sequence of nucleotides is then added to the corresponding 5' end (see  Figure~\ref{fig.daughter.cells}, and \cite{BlackburnCollins2011, Rhyu1995} for further details).

As discussed by \cite{TeixeiraArnericEtAl2004}, the number of nucleotides added varies between a few to more than 100 nucleotides, and the number added is independent of the length of the strand. However,  when telomerase is present, it does not lengthen every telomere in the cell (as illustrated in Figure~\ref{fig.daughter.cells}); it is more likely to elongate a telomere if the telomere is shorter (we also refer the reader to the telomere length regulation model proposed by~\cite{Greider2016}).

At the end of the replication process (including the possible elongation of telomeres), the resulting object is a pair of chromosomes, with one being given to each of the daughter cells.

\begin{figure}
\captionsetup[subfigure]{justification=centering}
\begin{center}
\begin{subfigure}{1\textwidth}
        \includegraphics[width=13.2cm, trim = 3cm 19.3cm 7.5cm 8.8cm, clip=true]{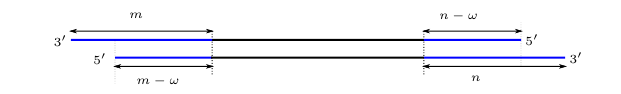}
        \caption{\label{fig.chrom.before.div}Chromosome before cell division.}
 \end{subfigure}
 \begin{subfigure}{1\textwidth}
        \includegraphics[width=13.2cm, trim = 4.4cm 16.3cm 6cm 10.5cm, clip=true]{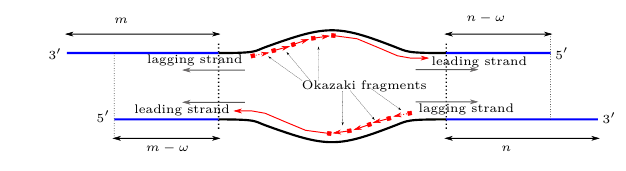}
        \caption{\label{fig.rep.fork}Replication bubble and replication forks}
 \end{subfigure}
 \begin{subfigure}{1\textwidth}
        \includegraphics[width=13.2cm, trim = 3cm 13.5cm 7.5cm 13.5cm, clip=true]{fig_tel_rep.pdf}
        \caption{\label{fig.repli.before.tel} The two daughter chromosomes before lengthening by telomerase}
 \end{subfigure}
  \begin{subfigure}{1\textwidth}
        \includegraphics[width=13.2cm, trim = 3cm 9.8cm 7.5cm 16.5cm, clip=true]{fig_tel_rep.pdf}
        \caption{\label{fig.daughter.cells}The two daughter chromosomes after lengthening by telomerase}
 \end{subfigure}
\end{center}
	\caption{\label{fig:scheme_bio}Complete reproduction of a chromosome with the new overhangs at each end.}
\end{figure}

\subsection{Branching model}\label{sec:branchingmodel}
We now propose a stochastic model for the evolution of a population of cells and their chromosomes based on the biological mechanism presented in the previous section. 
In particular, we will define a branching process that mimics the behaviour of a population of dividing cells as a function of their telomeres.

We first introduce some notation in order to keep track of the telomere lengths in each chromosome. Recall that the 3' end overhangs the opposing 5' end of a DNA strand. We denote by $\overhang$ \label{notation:overgang} the length of this overhang and assume that it is the same for all telomeres. We will also write $(m, n)$\label{notation:mn} to denote a chromosome whose 3' ends are of length $m$ and $n$ (classical units are either nucleotides (nt) for single strand regions or base pair (bp) and kilobase (kb) for double stranded regions, and, for sake of simplicity we will use kilobases throughout the paper). Using this notation, it follows that at one end of the chromosome the 3' end and its opposing 5' end are of length $m$ and $m-\overhang$, respectively, while at the other end of the chromosome, they are of length $n$ and $n-\overhang$, see Figure~\ref{fig.chrom.before.div}.

If telomerase is not expressed then, after replication, the chromosome $(m,n)$ gives rise to two daughter chromosomes $(m, n-\overhang)$ and $(m-\overhang,n)$ (see Figure~\ref{fig.repli.before.tel}). 
However, when telomerase is present, each of the four new telomeres can be elongated. Recall that if the 3' end of one strand is elongated then the complementary 5' end will also be elongated, keeping an overhang of size $\overhang$. We will assume that the four telomeres (two per daughter chromosome) are elongated independently of the other telomeres, and with a probability that depends on its length. Moreover, the size of the elongation is assumed to be independent of the length of the telomere.

To represent this mathematically, let $q(\ell)$\label{notation:qell} denote the probability that a telomere 
with length $\ell$ will be elongated, and let $\chi_\ell$\label{notation:chiell} denote a Bernoulli random variable with parameter $q(\ell)$. Then, given $\chi_\ell = 1$, we let $L$\label{notation:L} denote the (random) quantity added to the telomere, whose law is given by a probability distribution $\mu_f$ \label{notation:mu_f} on $\mathbb N=\{0,1,2,\ldots\}$. Writing $\Lambda_\ell = \chi_\ell L$, \label{notation:Lambda} the two resulting daughter chromosomes are thus given by $(m + \Lambda_m, n - \overhang + \Lambda'_{n-\overhang})$ and $(m - \overhang + \Lambda''_{m-\overhang}, n + \Lambda'''_{n})$ (see Figure~\ref{fig.daughter.cells}), where  $(\Lambda_{\ell})_{\ell\in\mathbb N}$, $(\Lambda_\ell')_{\ell\in\mathbb N}$, $(\Lambda_\ell'')_{\ell\in\mathbb N}$ and $(\Lambda_\ell''')_{\ell\in\mathbb N}$ are mutually independent and identically distributed. 
In the example of Figure~\ref{fig.daughter.cells}, $\Lambda_n'''=0$. 

\medskip

\begin{rem}
    According to the empirical evidence presented by~\cite{TeixeiraArnericEtAl2004}, we choose the distribution $\mu_f$ to be independent of the length of the telomere $\ell$. However, the mathematical results and proofs developed in the next section can be easily adapted to the more general setting where the distribution does depend on $\ell$.
\end{rem}

\medskip

Apart from their telomere lengths, cells are also represented by their state: the cell can be in an active or a non-active state. In an active state, it can undergo replication, deactivation (i.e. cellular senescence, defined as an irreversible arrest of cell proliferation \citep{HAYAT20173}), or be removed from the system (including in particular apoptosis). In a non-active state (i.e. after senescence), the cell cannot replicate, and can thus only be removed from the system.

\medskip

With this in mind, we now define our model at the level of the cells. Let $N_t$\label{notation:Nt} be the number of cells in the system at time $t \ge 0$. The collection of cells will be denoted by $\{(c_i(t), x_i(t)) : i = 1, \dots, N_t \}$, where $c_i(t)$ represents the lengths of the telomeres in the $i$-th cell at time $t$ and $x_i(t)$ represents its state (active or non-active). 
More precisely, letting $K$\label{notation:K} denote the number of chromosomes in a cell, we write $c_i(t) = (m_{i, j}(t), n_{i, j}(t))_{j = 1}^K\in (\mathbb N\times\mathbb N)^K$, where $(m_{i, j}(t), n_{i, j}(t))$ represents the $j^{th}$ chromosome in cell $i$ at time $t \ge 0$. Moreover, if the $i^{th}$ cell is active we set $x_i(t)=\nsen$,\label{notation:sennsen} and otherwise $x_i(t)=\sen$.
Note that $N_t$ also counts the number of non-active cells. 
We assume that if the minimum telomere length of a cell is smaller than a given value $L_{\min}\geq \overhang$\label{notation:Lmin}, then it is inactive.

For a given element $c = (m_j, n_j)_{j = 1}^K \in (\mathbb N\times\mathbb N)^K$, we let $\min c:=\min\{m_j,n_j,\ 1\leq j\leq K\}$ denote the minimum telomere length in the cell and $\max c:=\max\{m_j,n_j,\ 1\leq j\leq K\}$ the maximum telomere length in the cell.
We thus denote by $E=E_{\nsen} \cup E_{\sen}$ the set of possible `values' a cell can take, with $E_{\nsen}$ and $E_{\sen}$ denoting the set of values of the active cells and non-active cells, respectively, i.e.
\[
E_{\nsen}:=\left\{c\in (\mathbb N\times \mathbb N)^K\, :\, \min c\geq L_{\min}\right\}\times \{\nsen\}\text{ and }E_{\sen}:=(\mathbb N\times \mathbb N)^K\times \{\sen\}.
\]
Hence, a cell is represented by an element of the form $(c, x) \in E$.

\medskip

Let 
\[
X_t \coloneqq \sum_{i = 1}^{N_t}\delta_{(c_{i}(t),x_i(t))}, \quad t \ge 0,
\] 
denote the population of cells at time $t$, which is an element of $\mathcal M(E)$, the set of non-negative discrete measures on $E$.
The system evolves as a branching process so that, given their point of creation, cells evolve independently of each other
 according to the following dynamics. In what follows, and throughout the rest of the article, we use the term {\it rate} for a function $r_{\mathtt x}(s, c)$, $\mathtt x \in \{\sen, \nsen\}, s \ge 0, c \in E$, to mean there is a random time, $T$, whose distribution is given by ${\rm Pr}_{u, (c, \mathtt x)}(T > t) = {\rm e}^{-\int_u^t r_{\mathtt x}(s, c){\rm d}s}$, $0 \le u \le t$.
 
 An active cell $(c,\nsen) = \left((m_j, n_j)_{j = 1}^K,\nsen\right)\in E_{\nsen}$ will remain in the same state until one of the following events occur.
\begin{itemize}\itemsep1em
\item At rate $d_\nsen(t,c)$, at time $t \ge 0$,\label{notation:dx}  the cell  is removed from the system:
\[
    X_t = X_{t-}-\delta_{(c,\nsen)}.
\]
Biologically, this represents the event that the cell is removed from the system via e.g. apoptosis. 

\item At rate $s_\nsen(t,c)$\label{notation:snsen}, at time $t\geq 0$, the cell  becomes non-active, which means that $x$ switches from $\nsen$ to $\sen$:
\[
    X_t = X_{t-}-\delta_{(c,\nsen)}+\delta_{(c,\sen)},
\]
which represents the physical phenomena of cellular senescence. 
\item at rate $b_\nsen(t, c)$\label{notation:bnsen}, at time $t \geq 0$, the cell divides into two daughter cells given by, 
$$c^{1} :=c^1(c) = \{(m_j -  \overhang B_j + \Lambda_{m_j -\overhang B_j ,j}, n_j - \overhang(1-B_j) + \Lambda'_{n_j- \overhang(1-B_j),j})_{j = 1}^K\}$$ and 
$$c^{2} :=c^2(c)= \{(m_j -  \overhang(1-B_j) + \Lambda_{m_j-  \overhang(1-B_j),j}'', n_j - \overhang B_j + \Lambda_{n_j- \overhang B_j ,j}''')_{j = 1}^K\},$$ where $B_j$, $j = 1, \dots, K$ are a collection of independent Bernoulli random variables with parameter $1/2$, describing the allocation of daughter chromosomes in both daughter cells, and the $\Lambda_{\ell,j}$, $\Lambda_{\ell,j}'$, $\Lambda_{\ell,j}''$ and $\Lambda_{\ell,j}'''$ are independent (of each other and the $B_j$), with law described in the previous subsection. If $\min c^{i}< L_{\min}$, for $i\in\{1,2\}$, then the cell becomes non-active immediately, and if $\min c^{i}\geq  L_{\min}$, we assume that the cell is, as its mother cell, active. 
Hence, the term $\delta_{(c,\nsen)}$ of $X_{t-}$ is replaced by 
\[1_{\min c^{1}\geq  L_{\min}}\times\delta_{(c^1,\nsen)}+1_{\min c^{1}<  L_{\min}}\times\delta_{(c^1,\sen)}+1_{\min c^{2}\geq  L_{\min}}\times\delta_{(c^2,\nsen)}+1_{\min c^{2}<  L_{\min}}\times\delta_{(c^2,\sen)}.
\]
Physically, this describes the event of cell replication and the possible configurations of the daughter cells.
\end{itemize}

A  non-active cell $(c,\sen)\in E_{\sen}$ 
remains in the same state until it is removed from the system, at rate  $d_\sen(t,c)$, at time $t \ge 0$:
\[
    X_t = X_{t-}-\delta_{(c,\sen)}.
\]
Again, this is the mathematical representation of apoptosis or other mechanisms that may cause cells to `die'.

\section{Mean semigroup of the branching process}
\label{sec:semigroup}
\subsection{Evolution equation of the mean semigroup}
In this section, we consider the dynamics of the expectation of the branching process $X_t$. We thus define the expectation semigroup associated to this process. For $(c,x)\in E$ and a bounded measurable function $g:E\to\mathbb{R}$, set
\begin{equation}
\psi_{u,t}[g](c, x) \coloneqq \mathbb{E}_{u,(c, x)}\left[\langle g, X_t \rangle\right] \coloneqq \mathbb{E}_{u,(c, x)}\left[\sum_{i = 1}^{N_t}g(c_{i}(t), x_i(t))\right], \quad t \ge u \ge 0, 
\label{expsg}
\end{equation}
 where $\mathbb{E}_{u,(c, x)}$ is the expectation operator associated to the law $\mathbb{P}_{u,(c, x)}(\cdot):=\mathbb{P}(\cdot | X_u=\delta_{(c,x)})$, i.e. the law of the process whose initial population is composed of a single cell, $(c, x) \in E$ at time $u$.

 In order to understand the evolution of the mean semigroup $(\psi_{u, t})_{0 \le u \le t}$, we introduce the following assumption that ensures only a finite number of events can occur in finite time intervals.
\begin{ass}\label{as:ratebound}
The birth, deactivation and removal rates, $b_\nsen$, $s_\nsen$, $d_\nsen$ and $d_\sen$, are bounded.
\end{ass}

Based on the dynamics described in the previous section, we have the following proposition, which gives the evolution equation for $(\psi_{u,t})_{0\leq u \leq t}$. In what follows, we say that the semigroup $\psi$ is {\it bounded in time} if $\|\psi_{u,t}[g]\|_\infty<\infty$ for all positive bounded measurable functions $g$ and for all $0\leq u \leq t$.
\begin{prop}\label{expsemilem}
Let $g:E\to\mathbb R$ be a bounded measurable function, $0 \le u \le t$, $(c,x) = \left((m_j, n_j)_{j = 1}^K,x\right) \in E$. Under Assumption~\ref{as:ratebound}, the expectation semigroup $(\psi_{u,t})_{0\leq u\leq t}$ is the unique solution, that is bounded in time, to
\begin{multline}
\psi_{u,t}[g](c,x) = g(c,x) + \mathbf{1}_{\{x=\nsen\}}\int_u^t  b_\nsen(v, c)\left[\mathcal{F}[\psi_{v,t}[g]](c,\nsen) - \psi_{v,t}[g](c,\nsen)\right] {\rm d}v\\
+ \mathbf{1}_{\{x=\nsen\}} \int_u^t s_\nsen(v, c)\left[\psi_{v,t}[g](c,\sen) - \psi_{v,t}[g](c,\nsen)\right]- \int_u^t d_x(v, c)\psi_{v,t}[g](c,x){\rm d}v,
\label{eveq1}
\end{multline}
where the operator $\mathcal{F}$ is defined by
\begin{align}
\nonumber
\mathcal{F}[g](c,\nsen) = \mathbb{E}\Big[&\mathbf{1}_{\min c^1\geq L_{\min}}g(c^1,\nsen) + \mathbf{1}_{\min c^1< L_{\min}}g(c^1,\sen) 
\\ \label{branchop}
&\quad + \mathbf{1}_{\min c^2\geq L_{\min}}g(c^2,\nsen) + \mathbf{1}_{\min c^2< L_{\min}}g(c^2,\sen)\Big],
\end{align}
with
\begin{align*}
&c^1=(m_j -  \overhang B_j + \Lambda_{m_j-  \overhang B_j ,j}, n_j - \overhang(1-B_j) + \Lambda'_{n_j- \overhang(1-B_j) ,j})_{j = 1}^K,\\
&c^2=(m_j -  \overhang(1-B_j) + \Lambda_{m_j-  \overhang(1-B_j),j}'', n_j - \overhang B_j + \Lambda'''_{n_j- \overhang B_j,j })_{j = 1}^K,
\end{align*}
and where the expectation $\mathbb{E}$ is taken with respect to the law of the $B_j$, $\Lambda_{\ell,j}$, $\Lambda_{\ell,j}'$, $\Lambda_{\ell,j}''$ and $\Lambda_{\ell,j}'''$.
\end{prop}

\begin{proof}
The proof is similar to \cite[\S 6]{CoxHarrisEtAl2019}. Let us recall the main arguments.
By Assumption~\ref{as:ratebound}, $\psi$ is bounded in time. 
By conditioning $\psi_{u,t + t'}$ at time $t$ and applying the Markov property, it is a straightforward exercise to show that $(\psi_{u,t})_{t \ge u\geq 0}$ is a semigroup. To show that it satisfies \eqref{eveq1} first consider the case when $x = \nsen$. Splitting the expectation semigroup on the first event (branching, deactivation or removal), we have
\begin{align}
\psi_{u,t}[g](c, \nsen) 
&= g(c, \nsen){\rm e}^{-\int_u^t b_\nsen(v, c) + d_\nsen(v, c) + s_\nsen(v, c){\rm d}v} \notag\\
&+ \int_u^t b_\nsen(v, c){\rm e}^{-\int_u^v b_\nsen(w, c) + d_\nsen(w, c) + s_\nsen(w, c){\rm d}w}\mathcal{F}[\psi_{v,t}[g]](c, \nsen){\rm d}v \notag\\
& +  \int_u^t s_\nsen(v, c){\rm e}^{-\int_u^v b_\nsen(w, c) + d_\nsen(w, c) + s_\nsen(w, c){\rm d}w}\psi_{v,t}[g](c, \sen) {\rm d}v,
\label{NSeveqn}
\end{align}
where the second term follows from an application of the branching property and the strong  Markov property. 
Applying \ref{thm:Dynkin}, we obtain equation~\eqref{eveq1} for the case $x = \nsen$.

For the case where $x = \sen$, the cell can only remain non-active or be removed from the system. Thus splitting $\psi_{u,t}[g](c, \sen)$ on the first time the cell is removed from the system and again, applying \ref{thm:Dynkin}, we obtain
\begin{equation}
  \psi_{u,t}[g](c, \sen) = g(c, \sen) - \int_u^t d_\sen(v, c)\psi_{v,t}[g](c, \sen){\rm d}v,
  \label{Seveqn}
\end{equation}
that is equation \eqref{eveq1} for the case $x=\sen$.

\medskip

We now show that \eqref{eveq1} admits a unique solution. Suppose that we have two solutions bounded in time, $\psi^{(1)}$ and $\psi^{(2)}$, to \eqref{eveq1}. Then, denoting $\chi _{u,t}[g]= \psi^{(1)}_{u,t}[g] - \psi^{(2)}_{u,t}[g]$ and $\bar\chi_{u,t}[g] = \sup_{(c, x)}|\chi_{u,t}[g](c, x)|$, due to Assumption~\ref{as:ratebound}, we have
\begin{align}
  \big|\chi_{u,t}[g](c, x)\big| &= \left|\mathbf{1}_{\{x=\nsen\}}\int_u^t  b_\nsen(v, c)\left[\mathcal{F}[\chi_{v,t}[g]](c,\nsen) - \chi_{v,t}[g](c,\nsen)\right] {\rm d}v \right. \notag\\
&+ \left.\mathbf{1}_{\{x=\nsen\}} \int_u^t s_\nsen(v, c)\left[\chi_{v,t}[g](c,\sen) - \chi_{v,t}[g](c,\nsen)\right]- \int_u^t d_x(v, c)\chi_{v,t}[g](c,x){\rm d}v \right| \notag \\
&\le (3\|b_{\nsen}\|_\infty + 2\|s_{\nsen}\|_\infty + \|d_{\nsen}\|_\infty + \|d_{\sen}\|_\infty)\int_u^t \bar\chi_{v,t}[g]\,{\rm d}v.\notag
\end{align}
Fixing $t$ and setting $f(r)=\bar\chi_{t-r,t}[g]$, we obtain, for all $r\in [0,t]$
\begin{align*}
f(r) &\leq (3\|b_{\nsen}\|_\infty + 2\|s_{\nsen}\|_\infty + \|d_{\nsen}\|_\infty + \|d_{\sen}\|_\infty)\int_{t-r}^t \bar\chi_{v,t}[g]\,{\rm d}v\\
              &= (3\|b_{\nsen}\|_\infty + 2\|s_{\nsen}\|_\infty + \|d_{\nsen}\|_\infty + \|d_{\sen}\|_\infty)\int_{0}^r \bar\chi_{t-w,t}[g]\,{\rm d}w
 \\
 	&=(3\|b_{\nsen}\|_\infty + 2 \|s_{\nsen}\|_\infty + \|d_{\nsen}\|_\infty + \|d_{\sen}\|_\infty)\int_{0}^r f(w)\,{\rm d}w.
\end{align*}
Uniqueness now follows easily from Gr\"onwall's inequality and from $f(0)=\bar \chi_{t,t}[g]=0$.
\end{proof}

\subsection{Many-to-one}\label{sec:manytoone}
We now consider a many-to-one representation for the semigroup $(\psi_{u,t})_{t \ge u\ge 0}$. The many-to-one formula offers another representation for the first moment of the system of cells via a weighted jump process. The motivation for considering this second representation is two-fold. Firstly, as we shall see in the next section, it more easily allows us to characterise the long-term behaviour of the branching process: indeed, considering the long-term behaviour of a single (weighted) trajectory is much more tractable than that of the entire system of cells. Secondly, we will later simulate our model using interacting particle systems that are based on multiple copies of single trajectories.

{To this end,} consider the process $(\mathcal{C}_t, \mathcal{X}_t)_{t \ge 0}$ that evolves as a pure jump process in $E\cup \{\partial\}$, where $(\mathcal{C}_t, \mathcal{X}_t)=\partial$ means that the cell is removed from the system and does not evolve further (in particular, the point $\partial$ is absorbing). The dynamics of this process are as follows.

\begin{itemize}
    \item When in configuration $(c,\nsen)= ((m_j, n_j)_{j = 1}^K, \nsen) \in E_{\nsen}$, the process
\begin{itemize}
\item  jumps with rate $2b_\nsen(t, c)$, at time $t$, to $(c',x')$ with
\begin{align*}
c'& =(m_j -  \overhang B_j + \Lambda_{m_j -  \overhang B_j,j}, n_j - \overhang(1-B_j) + \Lambda'_{n_j- \overhang(1-B_j),j})_{j = 1}^K,\\
x' & = 
\begin{cases}
\nsen & \text{ if } \min c' \geq L_{\min},\\
\sen & \text{ if } \min c' < L_{\min};
\end{cases}
\end{align*}

    \item switches with rate $s_\nsen(t,c)$ from $(c,\nsen)$ to $(c,\sen)$;
    \item switches with rate $d_{\nsen}(t,c)$ from $(c,\nsen)$ to $\partial$. 
\end{itemize}
\item The process in a configuration  $(c, \sen)$  switches to $\partial$ with rate $d_{\sen}(t,c)$.
\end{itemize}
Note that for active cells, when a jump occurs at rate $2b_\nsen$, the cell jumps from $\mathcal{C}_{t-}$ to the cell $c^1$ defined as in~\eqref{branchop}. Note that as the law of $c^1$ and $c^2$ are the same, it is equivalent to choose the cell $c^1$ with probability $1/2$ and the cell $c^2$ with probability $1/2$. Then, if $\min c'\geq L_{min}$, the cell remains active ($\mathcal{X}_t = \nsen$), and otherwise the cell becomes inactive ($\mathcal X_t=\sen$).

\medskip

Let $\mathbf{P}_{u,(c, x)}$ denote the law of this process when started from a single cell with configuration $(c, x)$ at time $u$, with corresponding expectation operator $\mathbf{E}_{u,(c, x)}$. Similarly to the proof of Proposition~\ref{expsemilem}, it follows that the linear semigroup $(\varphi_{u,t})_{t \ge u\ge 0}$ associated to $(\mathcal{C}_t, \mathcal{X}_t)_{t \ge 0}$ satisfies
\begin{align}
\varphi_{u,t}[g](c,x) &\coloneqq \mathbf{E}_{u,(c,x)}\left[g(\mathcal{C}_t, \mathcal{X}_t)\mathbf{1}_{t<\tau} \right]\notag\\
 &= g(c,x) + \mathbf{1}_{x=\nsen} \int_u^t \,2b_\nsen(v, c)\left[\frac12{\mathcal{F}}[\varphi_{v,t}[g]](c,\nsen) - \varphi_{v,t}[g](c,\nsen)\right]{\rm d}v\notag\\
&\qquad +\mathbf{1}_{x=\nsen}\int_u^t  \,s_\nsen(v, c)\left[\varphi_{v,t}[g](c,\sen) - \varphi_{v,t}[g](c,\nsen)\right]{\rm d}v\notag\\
&\qquad -\int_u^t  \,d_x(v, c) \varphi_{v,t}[g](c,x){\rm d}v,
\label{onecell}
\end{align}
where $(c, x) \in E$, $g: E \to \mathbb{R}$ is a measurable, bounded function, $0 \le u \le t$ and $\tau$ denotes the hitting time of $\partial$ by the process $(\mathcal{C},\mathcal{X})$.

\medskip

The above semigroup describes the average behaviour of a typical particle in the branching process, $X$. However, this semigroup does not take into account mass creation and loss, as in the branching process. The following many-to-one formula shows one how to deal with this.

\smallskip

\begin{lem}
\label{lem:manytoone}
Under the assumptions of Proposition~\ref{expsemilem}, the semigroup, defined for all bounded measurable functions $g:E\to\mathbb R$, initial times $u$, initial cell configurations $(c,x) \in E$ and times $t \ge u$ by
\begin{equation}
\phi_{u,t}[g](c,x) \coloneqq \mathbf{E}_{u,(c,x)}\left[\exp\left({\int_u^t \mathbf{1}_{\{\mathcal{X}_v = \nsen\}}b_\nsen(v, \mathcal{C}_v) {\rm d}v}\right)g(\mathcal{C}_t, \mathcal{X}_t)\mathbf{1}_{t<\tau}\right],
\label{m21}
\end{equation}
 also solves equation~\eqref{eveq1}, and hence $\psi_{u,t} = \phi_{u,t}$ for each $t \ge u\ge 0$.
\end{lem}

\smallskip

\begin{proof}[Sketch proof]
    The method behind the proof is almost identical to that of Proposition \ref{expsemilem} and so we outline the steps and leave the details as an exercise for the reader. 
    
    Conditioning on the first event time, as outlined in the dynamics described above, yields an equation similar to that of \eqref{onecell} except that the first term will be weighted by ${\rm e}^{\int_u^t \mathbf{1}_{\{x = \nsen\}}b_\nsen(s, c) {\rm d}s}$, and each of the integrands will be weighted by ${\rm e}^{\int_v^t \mathbf{1}_{\{x = \nsen\}}b_\nsen(s, c) {\rm d}s}$. The conclusion now follows from an application of Theorem \ref{thm:Dynkin}.
\end{proof}

\smallskip

The main idea behind the many-to-one lemma is that since we are only interested in the average behaviour of the branching process, rather than the pathwise behaviour, we can interpret the first moment of the whole branching particle system as the average behaviour of a typical trajectory in the whole particle system, weighted by the number of particles we expect to see in the system. In this case, the process $(\mathcal C_t, \mathcal X_t)_{t \ge 0}$ follows the dynamics of a `typical trajectory' and the expectation of the exponential term in \eqref{m21} corresponds to the expected number of particles in the branching process at time $t$.

\section{Long-term behaviour}\label{sec:QSD}
This section is devoted to studying the asymptotic stability of the telomere length profile, as stated in Theorem~\ref{thm:mainresult} below. For simplicity, we only consider the time homogeneous dynamics (hence replacing $b_{\nsen}(t,c)$ by $b_{\nsen}(c)$, $d_{\nsen}(t,c)$ by $d_{\nsen}(c)$, $\psi_{s,t}$ by $\psi_{t-s}$ and so on). This will allow us to make a connection to the theory of quasi-stationary distributions. We refer the reader to the book of~\cite{DelMoral2004} and the articles of \cite{DelMoralMiclo2002,DelMoral2013,BansayeCloezEtAl2019,ChampagnatVillemonais2018,velleret2019exponential} for time inhomogeneous versions of this theory.

In order to state our main result, we need some additional assumptions.
We first make a technical assumption (Assumption~\ref{as:irreducibility}) on the telomere lengthening caused by telomerase, which will ensure that the process is irreducible in $E_\nsen$ (see Lemma~\ref{lem:irreducibility} below). 
The second assumption (Assumption~\ref{as:lambda0pos}) then imposes that the process is supercritical (in the usual sense, made precise below). Finally we make assumptions which guarantee quasi-compactness (Assumptions~\ref{as:moment} and~\ref{as:qgeom}). We also remind the reader that the table of notation included in Appendix \ref{sec:notation} may be of use. 

\begin{ass}\label{as:irreducibility}
    We assume that 
    \begin{enumerate}
        \item there exists $b_{tel}>\overhang$ in $\mathbb N$ such 
        that $\mu_f\{b_{tel}\}>0$ and $\mu_f\{b_{tel}-1\}>0$,
        \item $b_\nsen(c)>0$ for all $(c,\nsen)\in E_{\nsen}$,
        \item $0<q(i)<1$ for all $i\geq L_{\min}$.
    \end{enumerate} 
\end{ass}

We will soon see that under Assumptions~\ref{as:ratebound} and~\ref{as:irreducibility}, the process is irreducible on $E_\nsen$ (see Lemma~\ref{lem:irreducibility}) and, by construction, its complement is absorbing. Hence we can define 
\begin{align}
\label{eq.def.lambda0}
\lambda_0=\inf\left\{\lambda\in\mathbb R,\ \text{such that}\ \liminf_{t\to \infty} e^{-\lambda t} \psi_{t}[\mathbf 1_{F}](c,\nsen)<+\infty \right\},
\end{align}
which does not depend on $(c,\nsen)\in E_\nsen$ nor on the (arbitrary) non-empty finite set $F\subset E_\nsen$. The informal interpretation of $\lambda_0$ is that it measures the asymptotic exponential growth of the expected number of active cells in the population. We say that the process is {\it supercritical} when $\lambda_0 > 0$, {\it subcritical} when $\lambda_0 < 0$ and {\it critical} when $\lambda_0 = 0$. The value $\lambda_0$ is sometimes referred to as the Malthusian parameter. 

In the rest of this section, we focus on the supercritical case, which corresponds to exponential growth of the average number of active cells (such as expanding cancerous tumour \citep{blasco2005telomeres} or germinal cells \cite[Fig.~1]{HiyamaHiyama2007}).

\begin{ass}
    \label{as:lambda0pos}
    We have $\lambda_0>0$.
\end{ass}

Next, we introduce an assumption on the moments of $\mu_f$, the law of the length added to telomeres when telomerase acts.

\begin{ass}
    \label{as:moment}     
    We assume that there exists $\alpha_0>\frac{1}{\overhang}\ln\frac{\|b_{\nsen}\|_\infty}{\lambda_0}$ such that  $\mu_{f}$ admits an exponential moment of order $\alpha_0>0$, i.e.
    \begin{align*}
    \sum_{n\geq 0} \mu_f\{n\}\, \exp(\alpha_0 n)<+\infty.
    \end{align*}
\end{ass}

We conclude with an assumption on the probability that telomerase acts on a telomere, depending on its length.
\begin{ass}\label{as:qgeom}
    We assume that the probability of telomerase activity, $q(i)$, decreases to $0$ when $i\to+\infty$.  
\end{ass}

Note that Assumption~\ref{as:qgeom} is satisfied, in particular, 
when $q(i)$ decreases geometrically fast in $i$, as suggested in  the telomere length regulation model~\citep{Greider2016}. Indeed, the replication fork model described therein posits that telomerase travels with the lagging strand replication machinery, requiring it to stay bound to the fork through nucleosomes and telomere proteins, either of which can cause
dissociation of telomerase, until reaching the chromosome terminus for telomere extension. By linking telomerase activity to fork progression subject to the renewal of succeeding non-dissociation events along the telomere, this model suggests that the activity of telomerase occurs with a probability that decreases geometrically fast in the length of the telomere.

\medskip

We are now ready to state our main result.

\begin{thm}
	\label{thm:mainresult}
	  Suppose Assumptions~\ref{as:ratebound},~\ref{as:irreducibility},~\ref{as:lambda0pos},~\ref{as:moment} and~\ref{as:qgeom} hold. Then there exists a function $\eta:E\to\mathbb R$, positive on $E_\nsen$ and vanishing on $E_\sen$, a probability measure $\nu$ on $E$, and a function $V:E\to[1,\infty)$  such that, for all $g:E\to\mathbb R$ satisfying $|g|\leq V$, we have
	\begin{align*}
	\left|e^{-\lambda_0 t}\psi_t[g](c,x)-\eta(c,x)\nu[g]\right|\leq C e^{-\gamma t} V(c,x), \quad t\geq 0,\ (c,x)\in E,
	\end{align*}
	 where $C,\gamma$ are positive constants. In addition, one can choose $V$ such that, for some $k\geq 1,\alpha>0$,
	\begin{align*}
	V(c,\sen)=1\text{ and }V(c,\nsen)=\exp\left[\alpha\sum_{j=1}^K\big((m_j-k+\overhang)_+ + (n_j-k+\overhang)_+\big)\right].
	\end{align*}
\end{thm}

\medskip

We prove this Theorem in three steps. The first takes the form of a Lemma that shows that the process is irreducible on $E_\nsen$. The second gives the asymptotic behavior for the process restricted to the active particles (Proposition~\ref{prop:quasicomp1}). Finally the proof of Theorem~\ref{thm:mainresult} follows from Theorem~3.1 of~\cite{CV2021reducible}.

\medskip

\begin{lem}
	\label{lem:irreducibility}
	Under Assumptions~\ref{as:ratebound} and~\ref{as:irreducibility}, the set $E_\nsen$
	is irreducible for the process $(\mathcal C_t,\mathcal X_t)_{t\geq 0}$ defined in Section~\ref{sec:manytoone}.
\end{lem}

\begin{proof}
	Let $(c, x) = ((m_j, n_j)_{j = 1}^K, \nsen)\in E_\nsen$. Let us denote by $\tau_1<\tau_2$ the two first jump times of the process (where $\tau_2=+\infty$ if $(\mathcal C_{\tau_1}, \mathcal X_{\tau_1})=\partial$). Then, for any fixed $t_0>0$ and any $(A_1\times \cdots \times A_K)\times\{\nsen\}\in E_\nsen$, we have
	\begin{align}
	\mathbf{P}_{(c,x)}&\left((\mathcal C_{t_0},\mathcal X_{t_0})\in (A_1\times\cdots\times A_K)\times\{\nsen\}\right)\nonumber\\
	&\geq \mathbf P_{(c,x)}\left(\tau_1<t_0<\tau_2,\,(\mathcal C_{t_0},\mathcal X_{t_0})\in (A_1\times\cdots\times A_K)\times\{\nsen\}\right)\nonumber\\
	&\geq (1-e^{-t_0\,\inf (2b_\nsen+s_\nsen+d_\nsen)}) e^{-t_0\,\sup (2b_\nsen+s_\nsen+d_\nsen)}\frac{\inf 2b_{\nsen}}{\sup (2b_\nsen+s_\nsen+d_\nsen)}\nonumber\\
	&\qquad\qquad \times \,\mu_{(m_1,n_1)}(A_1)\times\cdots\times \mu_{(m_K,n_K)}(A_K),
	\label{eq:useful1}
	\end{align}
	where
	\[
	\mu_{(m_j,n_j)}(A_j)=\mathbb P\left((m_j-\overhang B_j+\Lambda_{m_j-\overhang B_j,j},n_j-\overhang(1-B_j)+\Lambda'_{n_j-\overhang(1-B_j),j})\in A_j\right).
	\]
	Note that due to the second part of Assumption~\ref{as:irreducibility}, $\inf 2b_{\nsen}>0$.
	
	We observe that, restricting to the events $\{B_j=1, \Lambda'_{n_j,j}=0\}$ (i.e., we have attrition $-\overhang$  
	on the left-hand side of the chromosome and telomerase does not act on the right-hand side) and 
	$\{B_j=0, \Lambda_{m_j,j}=0\}$, we obtain 
	\begin{align*}
	\mu_{(m_j,n_j)}(\cdot)&\geq \mathbb P(B_j=1, \Lambda'_{n_j,j}=0)\,\mathbb P\left((m_j-\overhang +\Lambda_{m_j,j},n_j)\in \cdot\right)\\
	&\qquad\qquad+\mathbb P(B_j=0, \Lambda_{m_j,j}=0) \,\mathbb P\left((m_j,n_j-\overhang +\Lambda'_{n_j,j})\in \cdot\right)\\
	&\geq C\,\left(\delta_{(m_j+b_{tel}-\overhang,n_j)}+\delta_{(m_j+b_{tel}-\overhang-1,n_j)}+\delta_{(m_j-\overhang,n_j)}\right)\\
	&\qquad\qquad+C\left(\delta_{(m_j,n_j+b_{tel}-\overhang)}+\delta_{(m_j,n_j+b_{tel}-\overhang-1)}+\delta_{(m_j,n_j-\overhang)}\right),
	\end{align*}
	where Assumption~\ref{as:irreducibility} ensures that $C$ can be chosen to be positive.
	
	Now consider the discrete time Markov process $Z^j$  evolving in $\mathbb N^2$ and with transition kernel
	\begin{align*}
	\frac{1}{6}\left(\delta_{(m_j+b_{tel}-\overhang,n_j)}+\delta_{(m_j+b_{tel}-\overhang-1,n_j)}+\delta_{(m_j-\overhang,n_j)}+\delta_{(m_j,n_j+b_{tel}-\overhang)}+\delta_{(m_j,n_j+b_{tel}-\overhang-1)}+\delta_{(m_j,n_j-\overhang)}\right).
	\end{align*}
	This process can jump with positive probability from $(m_j,n_j)$ to $(m_j+\overhang(b_{tel}-\overhang),n_j)$ in $\overhang$ steps and from $(m_j+\overhang(b_{tel}-\overhang),n_j)$ to  $(m_j,n_j)$ in $b_{tel}-\overhang$ steps, so that there is path of length $b_{tel}$ linking $(m_j,n_j)$ to itself. Similarly, there exists a path of length $b_{tel}-1$ linking $(m_j,n_j)$ to itself. This implies that the process is aperiodic. Moreover, the process can jump from $(m_j,n_j)$ to $(m_j+\overhang(b_{tel}-\overhang-1)+1,n_j)$ in $\overhang$ steps, and then come back to $(m_j+1,n_j)$ in $b_{tel}-\overhang-1$ steps. This shows that the process can reach $(m_j+1,n_j)$ from $(m_j,n_j)$. Similarly, one shows that the process can reach $(m_j-1,n_j)$, $(m_j,n_j+1)$, $(m_j,n_j-1)$ from $(m_j,n_j)$ in less than $b_{tel}$ steps. In particular, we deduce that, for all $(m_1,n_1),\ldots,(m_K,n_K)$ and $(m'_1,n'_1),\ldots,(m'_K,n'_K)$, there exists $n_0$ such that
	\begin{align*}
	\mathbb P(Z^1_{n_0}=(m'_1,n'_1),\ldots,Z^K_{n_0}=(m'_K,n'_K)\mid Z^1_{0}=(m_1,n_1),\ldots,Z^K_{0}=(m_K,n_K))>0,
	\end{align*}
	where the $Z^i$ are chosen independent.
	Using the definition of $\mu_{(m_j,n_j)}$ and the inequality~\eqref{eq:useful1}, we deduce that, for all $(c',\nsen)\in E_\nsen$,
	\begin{align*}
	\mathbf P_{(c,x)}\left(\mathcal C_{n_0\,t_0}=c',\mathcal X_{n_0\,t_0}=\nsen\right)>0.
	\end{align*}
\end{proof}

In the next proposition, we use the notation
\[
G:=\left\{c\in(\mathbb N\times\mathbb N)^K,\ \min c>L_{min}\right\}
\]
to denote the set of possible values of active cells.

\begin{prop}
    \label{prop:quasicomp1}
    Suppose Assumptions~\ref{as:ratebound},~\ref{as:irreducibility},~\ref{as:lambda0pos},~\ref{as:moment} and~\ref{as:qgeom} hold. Then there exists a positive function $\eta_\nsen:G\to\mathbb R$, a probability measure $\nu_{\nsen}$  on $G$, and a function $V_\nsen:G\to[1,\infty)$  such that, for all $g:G\to\mathbb R$ satisfying $|g|\leq V_\nsen$, 
    \begin{align*}
    \left|e^{-\lambda_0 t}\psi_t[g_0](c,\nsen)-\eta_\nsen(c)\nu_\nsen[g]\right|\leq C e^{-\gamma t} V_\nsen(c),\ \forall t\geq 0,\ \forall c\in G,
    \end{align*}
    where $g_0(c,x):=g(c)\mathbf 1_{c\in G, x=\nsen}$, and where $C,\gamma$ are positive constants. In addition, one can choose $V_\nsen$ such that, for some $k\geq 1,\alpha>0$,
    \begin{align*}
    V_\nsen(c)=\exp\left[\alpha\sum_{j=1}^K\big((m_j-k+\overhang)_+ + (n_j-k+\overhang)_+\big)\right].
    \end{align*}
\end{prop}

    \begin{rem}
        The above proposition is stated for $\lambda_0>0$ which corresponds to exponentially growing populations, which is our main focus. However, a straightforward adaptation of the proof shows that the result also holds if $\inf (d_{\nsen}+s_{\nsen}+\lambda_0)>0$. 
        In addition, numerical simulations  suggest that this result may hold in a more general context, at least for some choices of the model parameters and without restriction on $\lambda_0$ (see Remark~\ref{rem:generalization}).
    \end{rem}

\begin{proof}
    
\newcommand{\tC}{\widetilde{\mathcal C}}

In order to prove the proposition, we consider the semigroup on $L^\infty(G)$ (the space of bounded measurable functions on $G$), defined, for all $c\in G$ and all bounded measurable function $g:G\to\mathbb R$, by
\begin{align}
\nonumber
\widetilde \varphi_t[g](c)&:=e^{-\|b_{\nsen}\|_{\infty}t}\phi_t[g_0](c,\nsen)\\
\nonumber &= \mathbf{E}_{(c,\nsen)}\left[\exp\left(-{\int_0^t (\|b_{\nsen}\|_\infty-b_\nsen(\mathcal{C}_v)) {\rm d}v}\right) g_0(\mathcal{C}_t, \nsen)\mathbf{1}_{t<\tau}\right]\\
\label{eq.phit}
&=:\mathrm{E}_{c}\left[g(\tC_t)\mathbf{1}_{t<\widetilde{\tau}}\right],
\end{align}
 where $(\tC_t)_{t\geq 0}$ is a pure jump sub-Markov process on $G$ evolving as $\mathcal C$ under $\mathbf E_{c,\nsen}$ but sent to a cemetery point $\dagger\notin G$ at time
 \[
  \widetilde\tau := \inf\left\{t > 0 : \mathcal{C}_t \in (E_{\nsen})^c \right\} \wedge
  \inf\left\{t > 0 : \int_0^t (\|b_{\nsen}\|_\infty-b_{\nsen}(\mathcal{C}_v)){\rm d}v > \mathbf{e}\right\},
 \]
 where $\mathbf{e}$ is an independent rate $1$ exponential random variable (\textit{i.e.} $(\tC_t)_{t\geq 0}$ is sent to $\dagger$ from  when it reaches $\left(E_{\nsen}\right)^c=E_{\sen}\cup \{\partial\}$ and at an additional rate $\|b_{\nsen}\|_\infty-b_{\nsen}$).
 More formally, its infinitesimal generator for bounded measurable functions $g:G\to\mathbb R$ is given by 
\begin{align}
\nonumber
\mathcal L g(c)&=2b_{\nsen}(c)\left(\frac{1}{2}\mathcal F[g_0](c,\nsen)-g(c)\right)-(d_{\nsen}(c)+s_\nsen(c)+\|b_{\nsen}\|_\infty-b_{\nsen}(c))g(c)\\
\label{eq.generator_phi}
&=b_{\nsen}(c)\mathcal F[g_0](c,\nsen)-(d_{\nsen}(c)+s_\nsen(c)+\|b_{\nsen}\|_\infty+b_{\nsen}(c))g(c),
\end{align}
where $\mathcal F$, defined in~\eqref{branchop}, satisfies
\begin{align*}
\mathcal F[g_0](c,\nsen)
&=2\mathbb E\left[ \mathbf 1_{\min c^1\geq L_{\min}}g\left((m_j -  \overhang B_j + \Lambda_{m_j-  \overhang B_j,j}, n_j - \overhang(1-B_j) + \Lambda'_{n_j- \overhang(1-B_j) ,j})_{j = 1}^K\right)\right].
\end{align*}
The irreducibility of $\tC$ is a direct consequence of the irreducibility of $\mathcal C$ proved in Lemma~\ref{lem:irreducibility}. In particular, we can define
\begin{align*}
\widetilde \lambda_{0}=\inf\left\{\lambda\in\mathbb R,\ \text{such that}\ \liminf_{t\to \infty} e^{-\lambda t} \widetilde \varphi_{t}[\mathbf 1_H](c)<+\infty \right\},
\end{align*}
independently of $c\in G$ and of the finite set $H\subset G$. We observe, by Lemma~\ref{lem:manytoone}, \eqref{eq.def.lambda0} and \eqref{eq.phit}, that
\begin{align}
\label{eq.tildelambda0_lambda0}
\widetilde \lambda_0=\lambda_0-\|b_\nsen\|_\infty.
\end{align}

Our aim is to apply Theorem~5.1 of~\cite{ChampagnatVillemonais2017}. In order to do so, it is sufficient to find a Lyapunov type function, $V_\nsen:G\to[1,+\infty)$, such that 
\begin{align}
\label{eq:lyap}
\mathcal LV_\nsen\leq (\widetilde \lambda_0-\varepsilon)V_\nsen+C\mathbf 1_H,
\end{align}
 for some $\varepsilon>0$ and finite $H\subset G$. Beware that the convention for the sign of $\lambda_0$ is not the same in the reference. Indeed, the constant $\lambda_0$ therein refers to a \textit{death rate}, while in our case it refers to a \textit{growth rate}. Hence, one need to replace $\lambda_0$ by $-\lambda_0$ in Theorem~5.1 of~\cite{ChampagnatVillemonais2017} to reconcile with our setting.
 
 In order to find $V_{\nsen}$, we set $\alpha\in\left(\frac{1}{\overhang}\ln\frac{\|b_{\nsen}\|_\infty}{\lambda_0},\alpha_0\right)$, where $\alpha_0$ is given by Assumption~\ref{as:moment}. We look at functionals of the type
 \begin{align*}
 V_{k}(c):=\exp\left[\alpha\sum_{j=1}^K\big((m_j-k+\overhang)_+ + (n_j-k+\overhang)_+\big)\right], \quad k\geq 0,
 \end{align*}
with $c = (m_j, n_j)_{j = 1}^K\in G$.
We introduce the event
 \begin{align*}
A_k(c)&=
	\big\{\Lambda_{m_j-\overhang B_j,j}=0,\,\forall j\text{ s.t. }m_j>\nicefrac k2\big\}\cap\big\{\Lambda'_{n_j- \overhang(1-B_j),j}=0,\, \forall j\text{ s.t. }n_j>\nicefrac k2\big\}\\
	& \qquad
		\cap\big\{\Lambda_{m_j-\overhang B_j,j}\vee \Lambda'_{n_j- \overhang(1-B_j),j}<\nicefrac k2-\overhang,\,\forall j\big\},
\end{align*}
and its probability
\[
p_k(c)=\mathbb P\left(A_k(c)\right).
\]
It follows from Assumption~\ref{as:qgeom} that
for all $c\in E$ the probability of the complementary of $A_k(c)$ satisfies
\begin{align*}
\mathbb P\left(A_k(c)^c\right)
& \leq \sum_{j=1}^K \left[q(m_j-\omega B_j)1_{m_j>\nicefrac k2}
+ q(n_j-\omega (1-B_j))1_{n_j>\nicefrac k2}+ 2\,\mathbb P(L\geq \nicefrac k2-\omega)\right]
\\
& \xrightarrow[k\to+\infty]{} 0,
\end{align*}
then
\[
\inf_{c\in E} \,p_k(c)\xrightarrow[k\to+\infty]{} 1.
\]
Denoting by $c'=(m_j-\overhang B_j + \Lambda_{m_j-\overhang B_j,j}, n_j - \overhang(1-B_j) + \Lambda'_{n_j- \overhang(1-B_j) ,j})_{j = 1}^K$ the configuration of the process after a jump that doesn't lead to the removal of the particle from the system, we easily check that
\begin{align}
\label{eq.maj.Vkac'}
V_{k}(c')\leq \exp\left(\alpha\sum_j(\Lambda_{m_j-\overhang B_j,j}+\Lambda'_{n_j- \overhang(1-B_j),j})\right) V_{k}(c).
\end{align}
On the one hand, by H\"older's inequality, for any fixed $1<p<\frac{\alpha_0}{\alpha}$, we have
\begin{align}
\nonumber
\mathbb E\Big(\mathbf 1_{A_k(c)^c} \exp(& \alpha\sum_j(\Lambda_{m_j-\overhang B_j,j}+\Lambda'_{n_j- \overhang(1-B_j),j})\Big)
\\ \nonumber 
&\leq (1-p_k(c))^{1-\nicefrac1p}\,\left(\mathbb E\left(\exp(p\alpha\sum_j(\Lambda_{m_j-\overhang B_j,j}+\Lambda'_{n_j- \overhang(1-B_j),j})\right)\right)^{\nicefrac1p}\\
\label{eq.complem.Ak}
&\leq (1-p_k(c))^{1-\nicefrac1p}(\theta_{p\alpha})^{\nicefrac{2K}p},
\end{align}
where $\theta_{p\alpha}:=\mathbb E\left(\exp(p\alpha L)\right)= \sum_{n\geq 0} \mu_f\{n\}\, \exp(p\alpha n)<+\infty$ by Assumption~\ref{as:moment}.
On the other hand, conditioning on the event $A_k(c)$,  the left-hand telomere length $m_j'$ of the $j$-th chromosome of $c'$ satisfies
\[
m_j' -k + \overhang = m_j-\overhang B_j + \Lambda_{m_j-\overhang B_j,j} - k+\overhang \leq 
\begin{cases}
m_j-k+\overhang-\overhang B_j & \text{if $m_j>\nicefrac k2$}\\
0  		 & \text{if $m_j\leq\nicefrac k2$}
\end{cases}
\]
and similarly for the right-hand side. Then $V_{k}(c')\leq V_{k}(c)$,
and if $c$ is such that $\max c\geq k$, then with probability at least $1/2$ (for at least one chromosome side which leads to $\max c\geq k$), $V_{k}(c')\leq e^{-\alpha\overhang} V_{k}(c)$.
We then deduce from \eqref{eq.generator_phi} that, for all $c\in G$ such that $\max c\geq k$,
\begin{align*}
\mathcal L V_{k}(c) &\leq 2b_{\nsen}(c)\left(\frac{p_k(c)}{2}+\frac{p_k(c)}{2}e^{-\alpha\overhang} + (1-p_k(c))^{1-\nicefrac1p}(\theta_{p\alpha})^{\nicefrac{2K}p}\right)V_{k}(c)\\
&\hspace{3cm}-(d_{\nsen}(c)+s_\nsen(c)+\|b_{\nsen}\|_\infty+b_{\nsen}(c))V_{k}(c)\\
&\leq \left(b_{\nsen}(c)e^{-\alpha\overhang}+2b_{\nsen}(c) (1-p_k(c))^{1-\nicefrac1p}(\theta_{p\alpha})^{\nicefrac{2K}p}-d_{\nsen}(c)-s_\nsen(c)-\|b_{\nsen}\|_\infty\right)V_{k}(c)\\
&\leq \left(b_{\nsen}(c)e^{-\alpha\overhang}+2b_{\nsen}(c) (1-p_k(c))^{1-\nicefrac1p}(\theta_{p\alpha})^{\nicefrac{2K}p}-\|b_{\nsen}\|_\infty\right)V_{k}(c).
\end{align*}
As $\alpha$ has been chosen such that $\alpha>\frac{1}{\overhang}\ln\frac{\|b_{\nsen}\|_\infty}{\lambda_0}$ then
\[
\varepsilon:=\frac12(\|b_{\nsen}\|_\infty+\widetilde\lambda_0-\|b_\nsen\|_\infty e^{-\alpha\overhang})=\frac12(\lambda_0-\|b_\nsen\|_\infty e^{-\alpha\overhang})>0,
\]
where we used~\eqref{eq.tildelambda0_lambda0}, so that $\|b_{\nsen}\|_\infty-b_{\nsen}(c)e^{-\alpha\overhang}\geq 2\varepsilon -\widetilde\lambda_0$.
Moreover, as $\theta_{p\alpha}<\infty$, then choosing $k$ large enough such that $2b_{\nsen}(c)(1-p_k(c))^{1-\nicefrac1p}(\theta_{p\alpha})^{\nicefrac{2K}p}\leq \varepsilon$ for all $c\in G$, we deduce that, for all $c\in G$ such that $\max c\geq k$, 
\begin{align*}
\mathcal L V_{k}(c) 
\leq (\widetilde\lambda_0-\varepsilon)V_{k}(c).
\end{align*}
Now, choosing the finite set  $H:=\{c\in G,\ \max c\leq k\}$,  from \eqref{eq.maj.Vkac'} and \eqref{eq.generator_phi}, $\mathcal L V_{k}\leq V_{k}$ which is bounded in the compact set $H$. We deduce that there exists a constant $C>0$ such that~\eqref{eq:lyap} holds true for $V_\nsen=V_{k}$, which concludes the proof.
\end{proof}

\begin{proof}[Proof of Theorem~\ref{thm:mainresult}]
    In a similar manner to \eqref{eq.phit}, consider the pure jump sub-Markov process $(\bar{\mathcal C}_t,\bar{\mathcal X}_t)_{t\geq 0}$ on $E$ evolving as $(\mathcal C,\mathcal X)$ under $\mathbf E_{(c,x)}$, but is removed from the system at an additional rate $\|b_\nsen\|_\infty-\mathbf 1_{x=\nsen}\,b_{\nsen}(c,x)$. 
Denote by $\mathrm E_{(c,x)}$ the expectation operator associated to this process, and by $\bar \varphi_t$ the semigroup on $L^\infty(E)$ (the space of bounded measurable functions on $E$), defined, for all $(c,x)\in E$ and all bounded measurable functions $g:E\to\mathbb R$, by
\begin{align}
\nonumber
\bar \varphi_t[g](c,x)&:=e^{-\|b_{\nsen} t\|_{\infty}}\phi_t[g](c,x)=\mathrm E_{(c,x)}\left[g(\bar{\mathcal C}_t,\bar{\mathcal X}_t)\mathbf 1_{t<\bar\tau}\right],
\end{align}
where  $\bar\tau$ denotes the time at which the particle is removed from the system $(\bar{\mathcal C}_t,\bar{\mathcal X}_t)_{t\geq 0}$.
   Note that, before time $T_{\sen}:=\inf\{t\geq 0,\ \bar{\mathcal X}_t=\sen\}$, the process $\bar{\mathcal C}_t$ evolves as the process $\widetilde{\mathcal C}_t$ introduced in the proof of Proposition~\ref{prop:quasicomp1}. Combining this fact with \eqref{eq.phit}, \eqref{eq.tildelambda0_lambda0} and Proposition~\ref{prop:quasicomp1}, we have
    \begin{align*}
    \left|e^{-\widetilde\lambda_0 t}\mathrm E_{(c,x)}\left[g(\bar{\mathcal C}_t,\bar{\mathcal X}_t)\mathbf 1_{t<\bar\tau\wedge T_\sen}\right]-\eta_\nsen(c)\nu_\nsen[g(\cdot,\nsen)]\right|\leq C e^{-\gamma t} V_\nsen(c),\ \forall t\geq 0,\ \forall (c,x)\in E_\nsen.
    \end{align*}
    Next, we make use of Theorem~3.1 of~\cite{CV2021reducible} for the discrete time process $(\bar{\mathcal C}_n,\bar{\mathcal X}_n)_{n\in\mathbb N}$, in the situation corresponding to Assumption~A1 therein. In order to do so, we set $D_1=E_\nsen$, $D_2=E_\sen$, $j_{0,P}=0$, $\theta_{0,P}=e^{-\widetilde\lambda_0}$, $\gamma=e^{-\|b_\nsen\|_\infty}<\theta_{0,P}$,  $c_1=1$, $W_R\equiv \mathbf 1_{(c,x)\in D_2}$ and $W_P=V_\nsen\mathbf 1_{(c,x)\in D_1}$. We deduce that for $g:E\to\mathbb R$ such that $|g|\leq V:=W_P+W_R$,
    \begin{align}
    \label{eq:discretetime}
    \left|e^{-\widetilde \lambda_0 n}\mathrm E_{(c,x)}\left[g(\bar{\mathcal C}_n,\bar{\mathcal X}_n)\mathbf 1_{n<\bar\tau}\right]-\eta(c,x)\nu[g]\right|\leq C\beta^n V(c,x), \quad \forall (c,x)\in E, \forall n\in \mathbb N,
    \end{align}
    for some $\beta\in(0,1)$. Now fix $h\geq 0$ and consider 
    \[
   g_h:(c,x)\mapsto \mathrm E_{(c,x)}\left[g(\bar{\mathcal C}_h,\bar{\mathcal X}_h)\mathbf 1_{h<\bar\tau}\right].
    \]
Applying~\eqref{eq:discretetime} to $g_h$ and using the Markov property at time $h$, we obtain
    \[
    \left|e^{-\widetilde \lambda_0 n}\mathrm E_{(c,x)}\left[g(\bar{\mathcal C}_{n+h},\bar{\mathcal X}_{n+h})\mathbf 1_{n+h<\bar\tau}\right]-\eta(c,x)\nu[g_h]\right|\leq C\beta^n  V(c,x)\xrightarrow[n\to+\infty]{} 0.
    \]
    Taking $(c,x)\in E_\nsen$, we thus observe that $\nu$ is a quasi-limiting distribution and hence a quasi-stationary distribution for the process (this is a classical result from the theory of quasi-stationary distributions, see for instance the works of ~\cite{MeleardVillemonais2012,ColletMartinezEtAl2013,DoornPollett2013}). In particular, there exists $\lambda\leq 0$ such that $\mathrm P_\nu\left((\bar{\mathcal C}_t,\bar{\mathcal X}_t)\in\cdot ,\ t<\bar\tau\right)=e^{\lambda t}\nu(\cdot)$, for all $t\geq 0$. The above convergence shows that $\lambda=\widetilde \lambda_0$ and we deduce that $\nu[g_h]=e^{\widetilde \lambda_0 h}\nu[g]$. Finally, we proved that
   \[
   \left|e^{-\widetilde \lambda_0 n}\mathrm E_{(c,x)}\left[g(\bar{\mathcal C}_{n+h},\bar{\mathcal X}_{n+h})\mathbf 1_{n+h<\bar\tau}\right]-\eta(c,x)e^{\widetilde \lambda_0 h}\nu[g]\right|\leq C\beta^n  V(c,x).
   \]
   Up to a change in the constant $C$ and setting $\gamma=\ln\nicefrac1\beta$, this concludes the proof of the theorem.
\end{proof}

\section{Numerical simulations}\label{sec:simulations}
This section is dedicated to the numerical exploration of the theoretical results of the previous section. 
In Section~\ref{sec:num_wo_telo}, we consider the model without telomerase (i.e. with $q\equiv 0$) and study its limiting population size, through the so-called Hayflick limit, as a function of the parameters of the model.
In Section~\ref{sec:num_with_telo}, we consider the model with telomerase and study, first the values of the Malthusian parameter $\lambda_0$, defined by \eqref{eq.def.lambda0}, as a function of the parameters of the model, and second the convergence of the process as stated by Theorem~\ref{thm:mainresult} when $\lambda_0>0$. 
Notice that, contrary to the Hayflick limit which takes transitory effects into account, $\lambda_0$ and the limit of the process only describe the asymptotic average behaviour, and do not depend on the initial condition.

Note that under Assumption \ref{as:lambda0pos}, the expected number of particles grows exponentially and thus a na\"ive numerical simulation of the process is not appropriate. Instead, we 
use an interacting particle approximation scheme based on genetic algorithms, described in Algorithm~\ref{algo.ibm} in appendix, which allows us to replace the mean semigroup by the average behaviour of a fixed size population. We refer the reader to the articles of~\cite{DelMoral2004,DelMoral2013} for detailed results and methods allowing to deduce the average dynamics of the population (including in particular $t\mapsto \mathbb E_{(c_0,\nsen)}(N_t)$ as well as the Malthusian parameter $\lambda_0$) from simulations of a population with fixed size $N$.

For simulations presented in Sections~\ref{sec:num_wo_telo} and \ref{sec:num_with_telo}, we compute the expectation of $100$ runs of Algorithm~\ref{algo.ibm} with $N=10000$ particles.  
Where possible, we have chosen parameter values that are consistent with empirical data or existing literature. In particular, the experiments described by~\cite{Chow:2012b4f} suggest that $\overhang$ should be chosen in the range $[12{\rm b},300{\rm b}]$. However, when information on possible parameter values is not available, we have chosen the values arbitrarily.

We remind the reader that we have included a table of notation in Appendix \ref{sec:notation} containing a description of the model parameters and a reference to where they were first introduced in this article. 

\subsection{Without telomerase}
\label{sec:num_wo_telo}
In this section, we consider the model without telomerase. As such, this model describes the dynamics of a population of cells whose telomere lengths can only decrease with time.
More precisely, we choose $q\equiv 0$ and by default \label{notation:c0}
\begin{align}
\label{eq:param1}
L_{min}=2\,\kb,\ \overhang=0.2\,\kb,\ K=46,\ c_0=\prod_{i=1}^K(8\kb,8\kb),
\end{align}
so that all the telomeres in the population have initial length  $8\kb$ (the unit kb refers to $1$ kilobase, sometimes written as kbp for kilobase pair), and\label{notation:rs}
\begin{align}
\label{eq:param2}
s_{\nsen}(c)=e^{-{(\min c-L_{min})}/{r_s}}\text{ with $r_s=100$},\ b_{\nsen}(c)=1,\  d_{\sen}(c)=d_{\nsen}(c)=0.
\end{align}
In particular, when all the cells have (at least) one telomere shorter than $L_{min}$, they are all in a non-active state, and the 
population can no longer evolve.

 One quantity of interest is the final size of the population, which we denote by $N_\infty$\label{notation:Ninfty}.
The Hayflick limit~\citep{HayflickMoorhead1961,ShayWright2000} of a cell population describes the total number of doubling of the population. Mathematically, we define this limit as\label{notation:H}
\begin{align}
\label{eq:hayflick}
H=\log_2 \mathbb E_{(c_0,\nsen)}\left[N_\infty\right],
\end{align}
where $c_0$ describes the telomere lengths of an initial cell. In what follows, we will investigate (numerically) the dependence of the Hayflick limit on certain model parameters, namely the number of chromosomes $K$, the minimal telomere length for active cells $L_{min}$, the overhang length $\overhang$, and the parameter $r_s$ of the ``deactivation'' rate $s_{\nsen}$.

\begin{rem}
    In biological experiments, the Hayflick limit is better described by $\log_2 \frac{N_\infty}{N_0}$. However, under mild assumptions, $\frac{N_\infty}{N_0}$ can be approximated by $\mathbb E\left(\frac{N_\infty}{N_0}\right)$ when the initial size of the population $N_0$ is large (this is a consequence of the law of large numbers and the branching property in our model). In addition, if all the initial cells are in state $(c_0,\nsen)$, $\mathbb E\left(\frac{N_\infty}{N_0}\right)=\mathbb E_{(c_0,\nsen)}\left[N_\infty\right]$.  This justifies our mathematical definition of the Hayflick limit $H$.
\end{rem}

\medskip

We first present the evolution of the expected population size for the above choice of parameters, see Figure~\ref{fig:pop_dyn1}. In the initial population, all telomere lengths are sufficiently large so that the non-active population is negligible. From \eqref{eq:param2}, the expected population size grows exponentially with rate $b_{\nsen}(c)=1$ (see the second figure in Figure \ref{fig:pop_dyn1} where the total expected population size follows the red dotted line until $t\approx 19$).  Then the growth slows down until the population stabilises at $\mathbb E_{(c_0,\nsen)}(N_\infty)$ (see figure on the left), which corresponds to the time when the population of active cells vanishes (see figure on the right).

\begin{figure}
	\includegraphics[height=4.3cm, trim = 0cm 0cm 0cm 0.cm, clip=true]{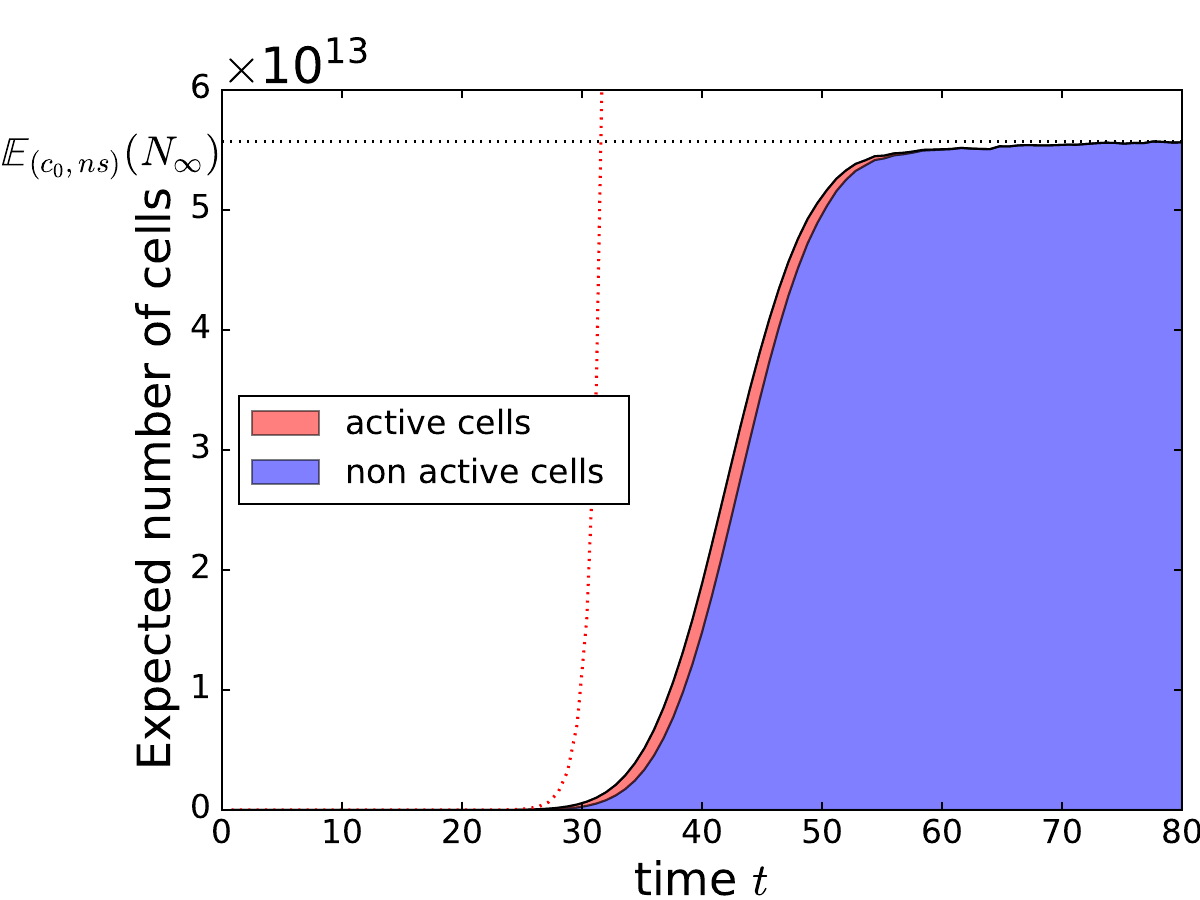}
    \includegraphics[height=4.3cm, trim = 0cm 0cm 1.7cm 0.cm, clip=true]{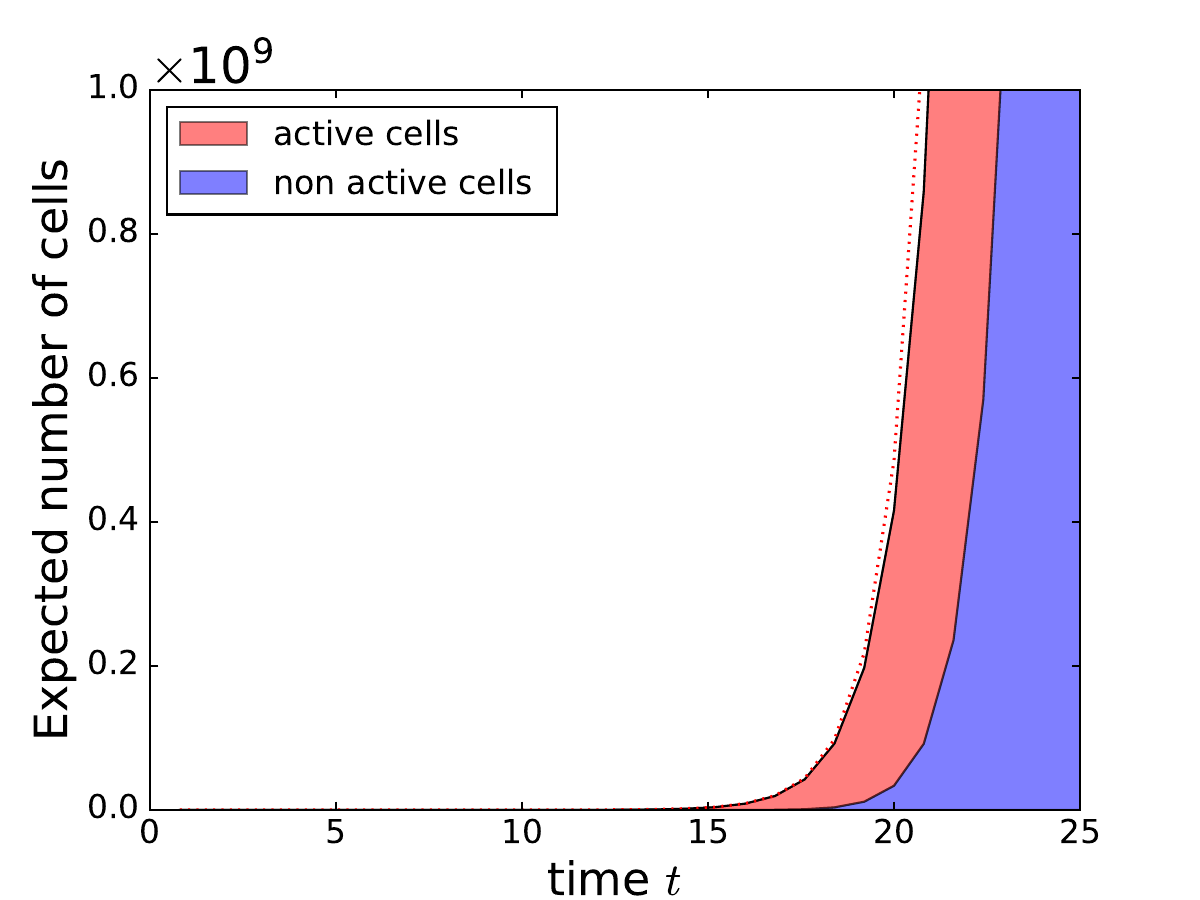}
    \includegraphics[height=4.3cm, trim = 0cm 0cm 1.7cm 0.cm, clip=true]{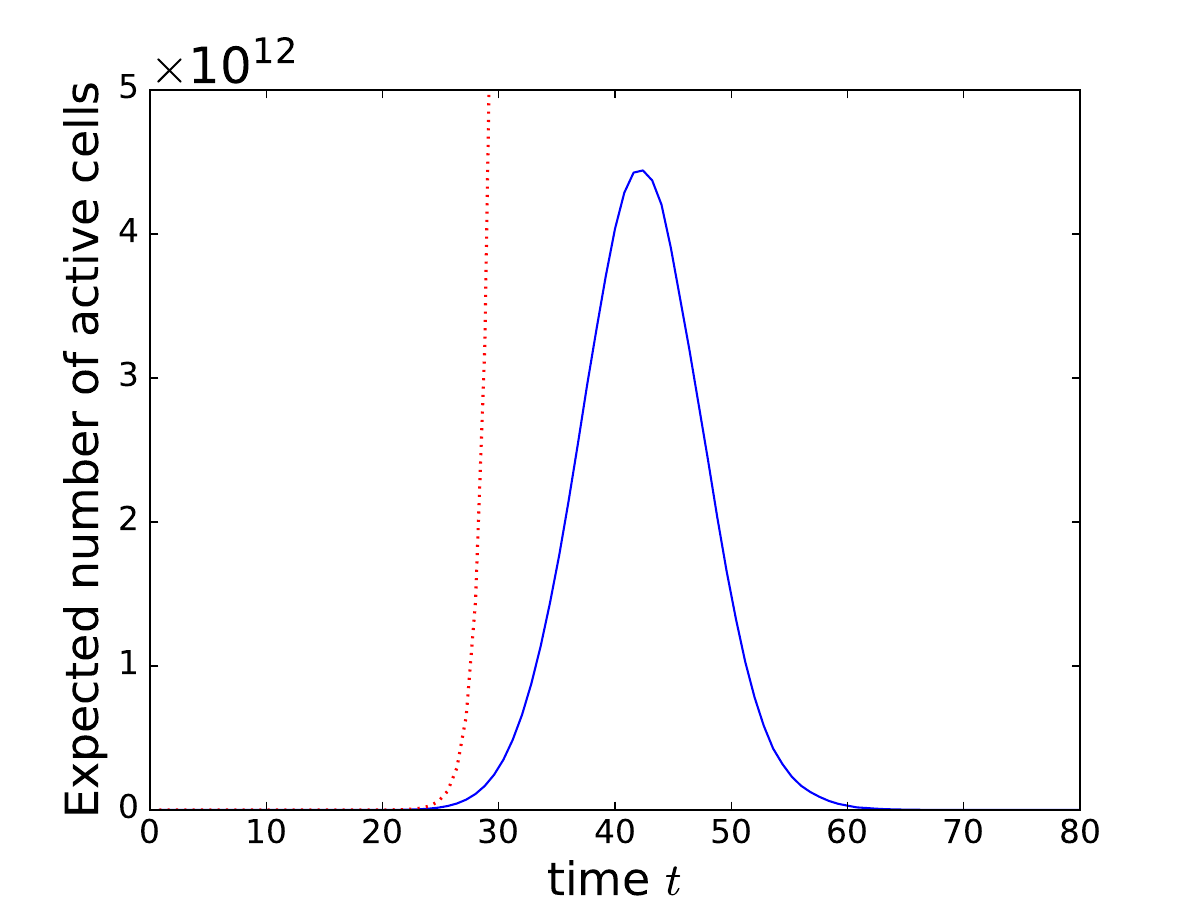}
\caption{\label{fig:pop_dyn1} Expected population size as a function of time $t\geq 0$, for the model without telomerase ($q\equiv0$) with parameters as in~\eqref{eq:param1} and~\eqref{eq:param2}. Left: the expected total population dynamics $\mathbb E_{(c_0,\nsen)}(N_t)$, with the proportion of active cells $\mathbb E_{(c_0,\nsen)}(X_t(E_\nsen))$ (red) and non active cells $\mathbb E_{(c_0,\nsen)}(X_t(E_\sen))$ (blue). Centre: zoom of the start of the dynamics in the left figure. Right: the dynamics of the expected number of active cells $\mathbb E_{(c_0,\nsen)}(X_t(E_\nsen))$. In each of the three graphs, the dotted red line is the curve of the function $t\mapsto e^{t}$.
}
\end{figure}

In Figure~\ref{fig:HayFctKLminOmegarS}, we represent the Hayflick limit as a function of the number of chromosomes $K$; the minimal length, $L_{min}$, of telomeres in active cells; the overhang $\overhang$, and the parameter $r_s$ associated with the ``deactivation'' rate $s_\nsen$.  We make the following observations.
\begin{itemize}
\item In Figure~\ref{fig:HayFctK}, we observe that the Hayflick limit is decreasing with respect to the number of chromosomes $K$. This is not surprising since one expects that the larger the number of chromosomes, the smaller (in law) the minimal telomere length in the cell. Therefore increasing the number of chromosomes increases the probability of deactivation (with our choice of $s_{\nsen}$ given by \eqref{eq:param2}), as well as  the probability that a cell is `non-active' after a division. 

With our choice for $L_{min}$, $\overhang$ and $c_0$ given in \eqref{eq:param1}, the minimal theoretical value of the Hayflick limit (when $K \to \infty$) is $(8000-2000)/200=30$, however we observe that it remains around $45$ for a large but realistic number of chromosomes. For the case with 46~chromosomes, which corresponds to the setting of Figure~\ref{fig:pop_dyn1}, we observe a Hayflick limit of $45.9$, which is comparable to the experimentally measured interval $[40,60]$ (see~\cite{HayflickMoorhead1961} and~\cite{VermaVermaEtAl2020}).
\item In Figure~\ref{fig:HayFctLmin}, we observe that the Hayflick limit decreases linearly with respect to the minimal length $L_{min}$ of active cells. This is due to the fact that increasing the minimal telomere length of active cells increases the number of non-active cells and hence decreases the number of times the population doubles.  
\item Similarly, the Hayflick limit is also decreasing with respect to the overhang $\overhang
$ (see Figure~\ref{fig:HayFctomega}) as well as the parameter $r_s$ of the deactivation rate, $s_\nsen$ (see Figure~\ref{fig:HayFctr_s}).
\end{itemize}

\begin{figure}
\begin{subfigure}{0.45\textwidth}
            \includegraphics[width=6cm]{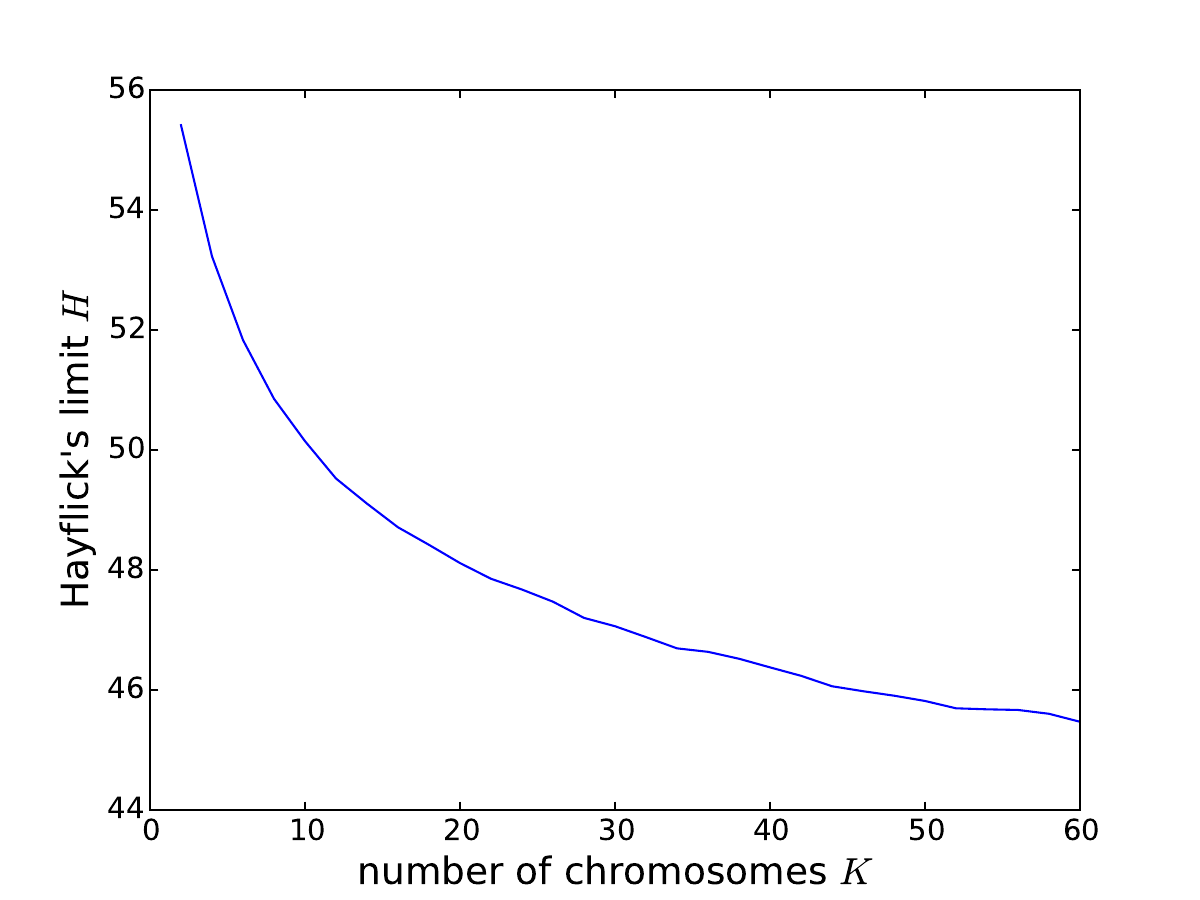}
            \caption{\label{fig:HayFctK}Hayflick limit as a function of $K$.}
        \end{subfigure}
        \begin{subfigure}{0.45\textwidth}
            \includegraphics[width=6cm]{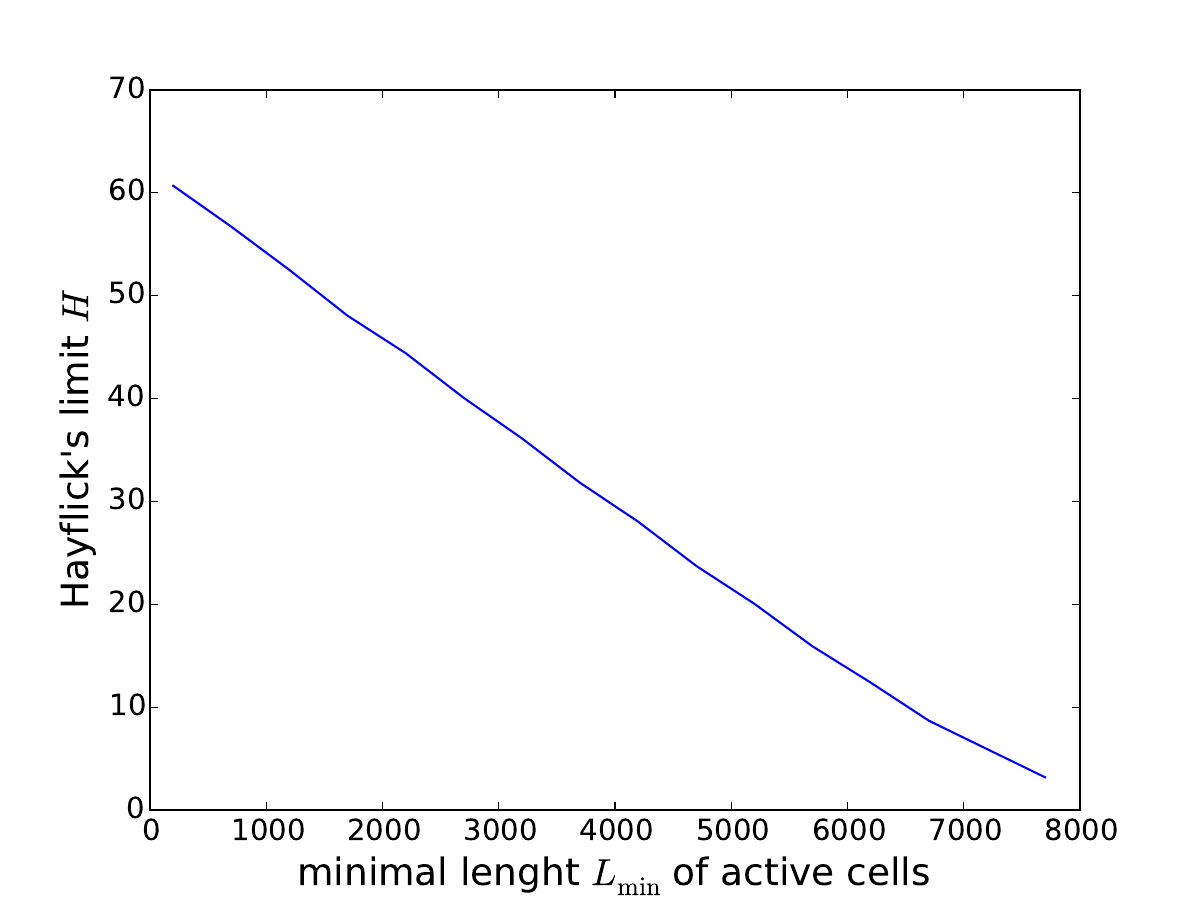}
            \caption{\label{fig:HayFctLmin}Hayflick limit as a function of $L_{min}$.}
        \end{subfigure}\\
        \begin{subfigure}{0.45\textwidth}
            \includegraphics[width=6cm]{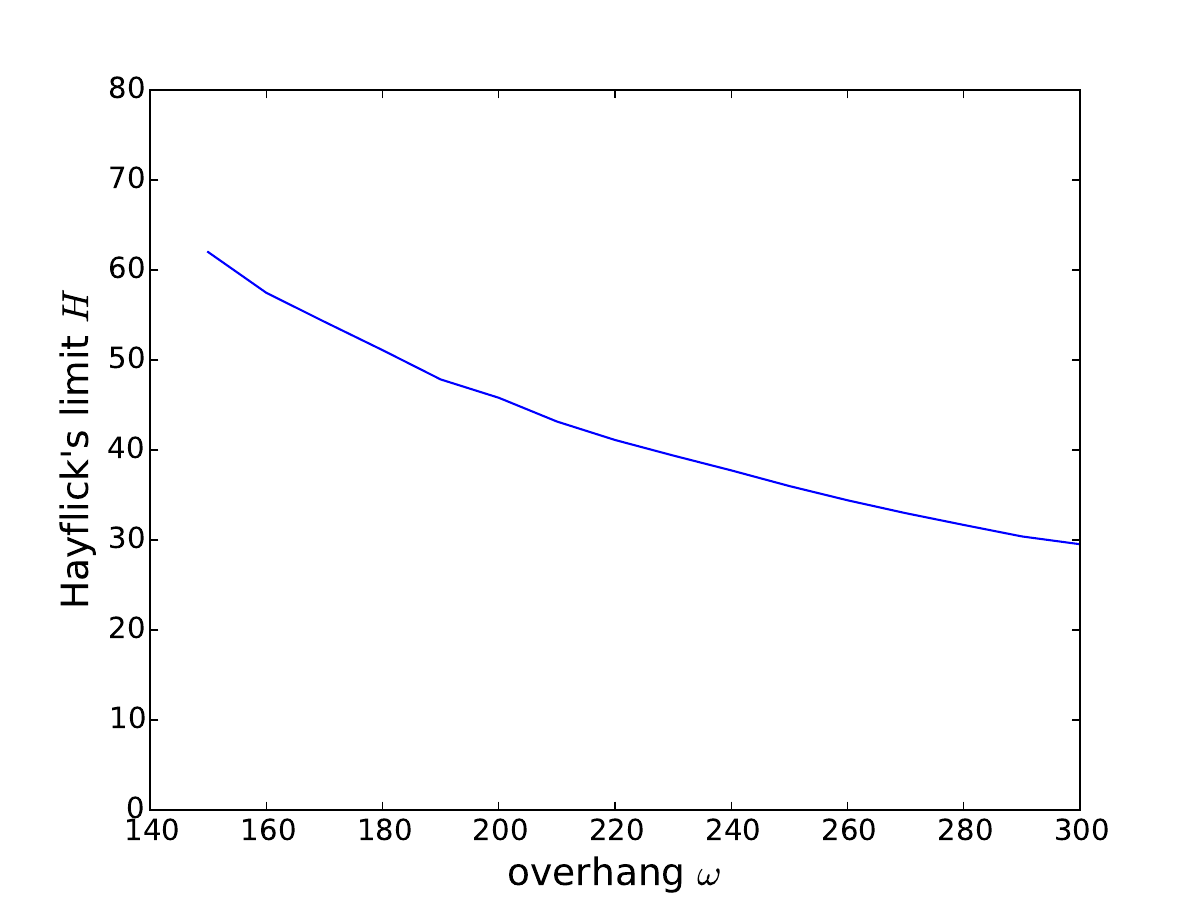}
            \caption{\label{fig:HayFctomega}Hayflick limit as a function of $\overhang$.}
        \end{subfigure}
        \begin{subfigure}{0.45\textwidth}
            \includegraphics[width=6cm]{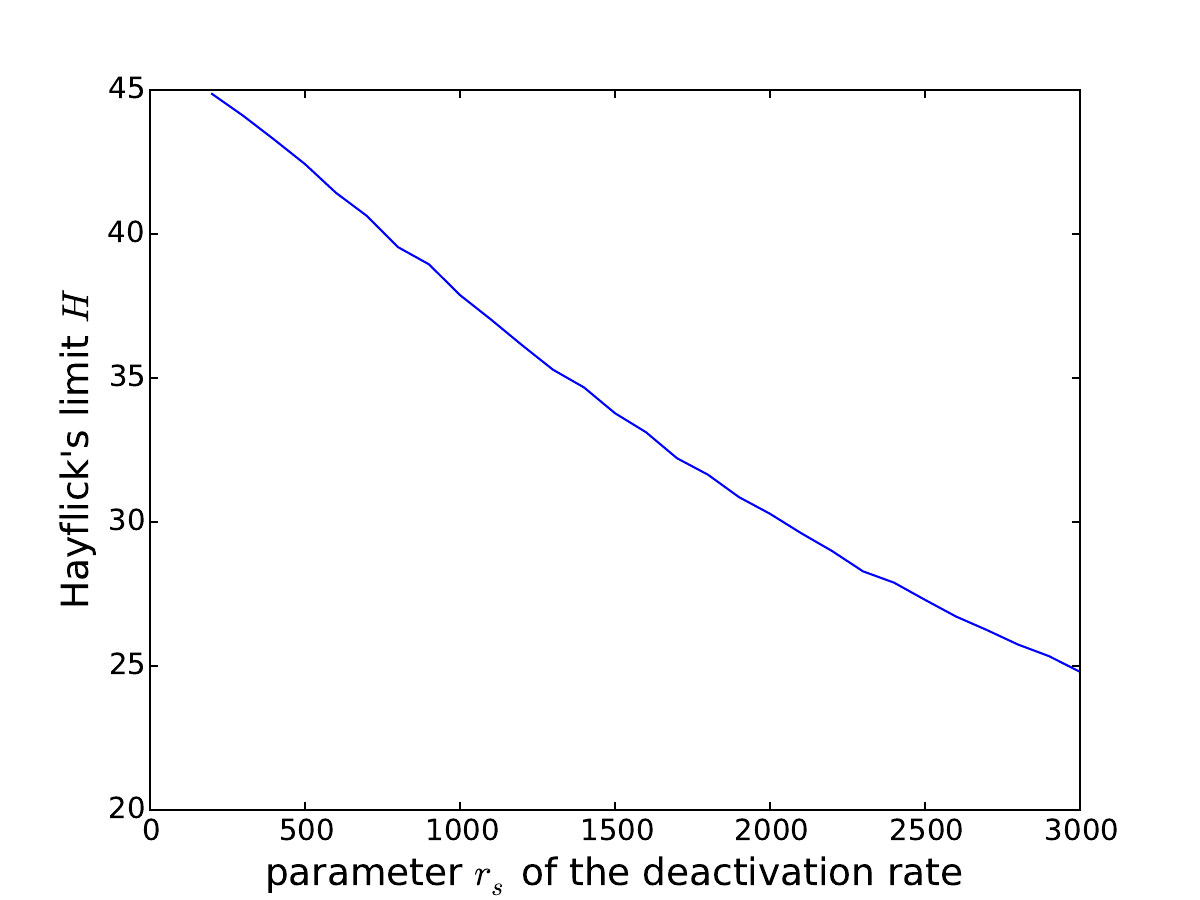}
            \caption{\label{fig:HayFctr_s}Hayflick limit as a function of $r_s$.}
        \end{subfigure}
    \caption{\label{fig:HayFctKLminOmegarS}  Hayflick limit~\eqref{eq:hayflick} for different values of $K$, $L_{min}$,  $\overhang$ and $r_s$,  all other parameters being as in~\eqref{eq:param1} and~\eqref{eq:param2}.  Note that,  the numerical approximation method used in Algorithm~\ref{algo.ibm} implies some uncertainties in the computation of the Hayflick limit, which are slightly visible, mainly in Figures~\ref{fig:HayFctK} and~\ref{fig:HayFctr_s}. }
\end{figure}

\subsection{With telomerase}
\label{sec:num_with_telo}
We consider now the model {\it with} telomerase and we choose default parameters satisfying Assumptions~\ref{as:irreducibility}, \ref {as:moment} and \ref{as:qgeom}. Our values of $q(\ell)$ are derived from the empirical measures described by~\cite{TeixeiraArnericEtAl2004} and the mechanism described by~\cite{Greider2016}. We choose the following geometric distribution for the probability that telomerase lengthens a telomere, and the following uniform distribution for the distribution of the size of the increase:
\begin{align}
\label{eq:param1telo}
q(\ell)=q_0\,2^{-(\ell-L_{min})/L_{min}}\text{ and }\mu_f=\mathcal U([0,M]),
\end{align}
where $q_0\in [0,1]$ and $M>0$.  By default, we choose $q_0=0.7$, $M=600$ and the other parameters remain as in~\eqref{eq:param1} and~\eqref{eq:param2}.

If $\lambda_0>0$ then the expected population size goes to infinity at rate $\lambda_0$ and hence the Hayflick limit is equal to infinity. On the other hand, if $\lambda_0<0$, then the asymptotic expected population size of active cells decreases exponentially fast and eventually goes extinct. Finally, if $\lambda_0=0$, then the population of active cells is critical and does not asymptotically increase nor decrease exponentially fast.

In order to understand the effect of the biological parameters on the behaviour of the population (survival, growth rate, Hayflick limit), we compute $\lambda_0$ for different  values of $s_{\nsen}$, $q$ and $\mu_f$ (equivalently, $r_s$, $q_0$ and $M$, respectively). When $\lambda_0< 0$, we compute the Hayflick limit and, when $\lambda_0>0$, we illustrate the convergence of  the telomere length distribution in the population of cells to the limiting distribution $\nu$, as stated in Theorem~\ref{thm:mainresult}.

\medskip

In Figure~\ref{fig:Lambda0_wrt_para}, we plot $\lambda_0$ as a function of the parameters $q_0$, $M$ and $r_s$. We make the following observations.
\begin{itemize}
\item As expected, the eigenvalue $\lambda_0$ increases when the parameter $q_0$ increases, since this increases then the probability of telomerase lengthening a telomere.
\item Similarly, when $M$ increases, so does $\lambda_0$ since this increases the upper bound on the amount a telomere can be lengthened by.
\item In the left-hand figure, we see that $\lambda_0$ decreases monotonically with $r_s$, the parameter that appears in the deactivation rate $s_\nsen$. The reasoning is the same as for the Hayflick limit. In fact, as we see in Figure~\ref{fig:HayFctr_s} for the model without telomerase,  increasing the parameters $r_s$ is not beneficial for the growth of the population (seen in terms of the number of times the population doubles). Similarly, in this case, it is not beneficial when the telomerase is active (seen here in terms of the population size growth rate $\lambda_0$). 
\end{itemize}

\begin{figure}
	\begin{subfigure}{0.32\textwidth}
            \includegraphics[width=5.5cm]{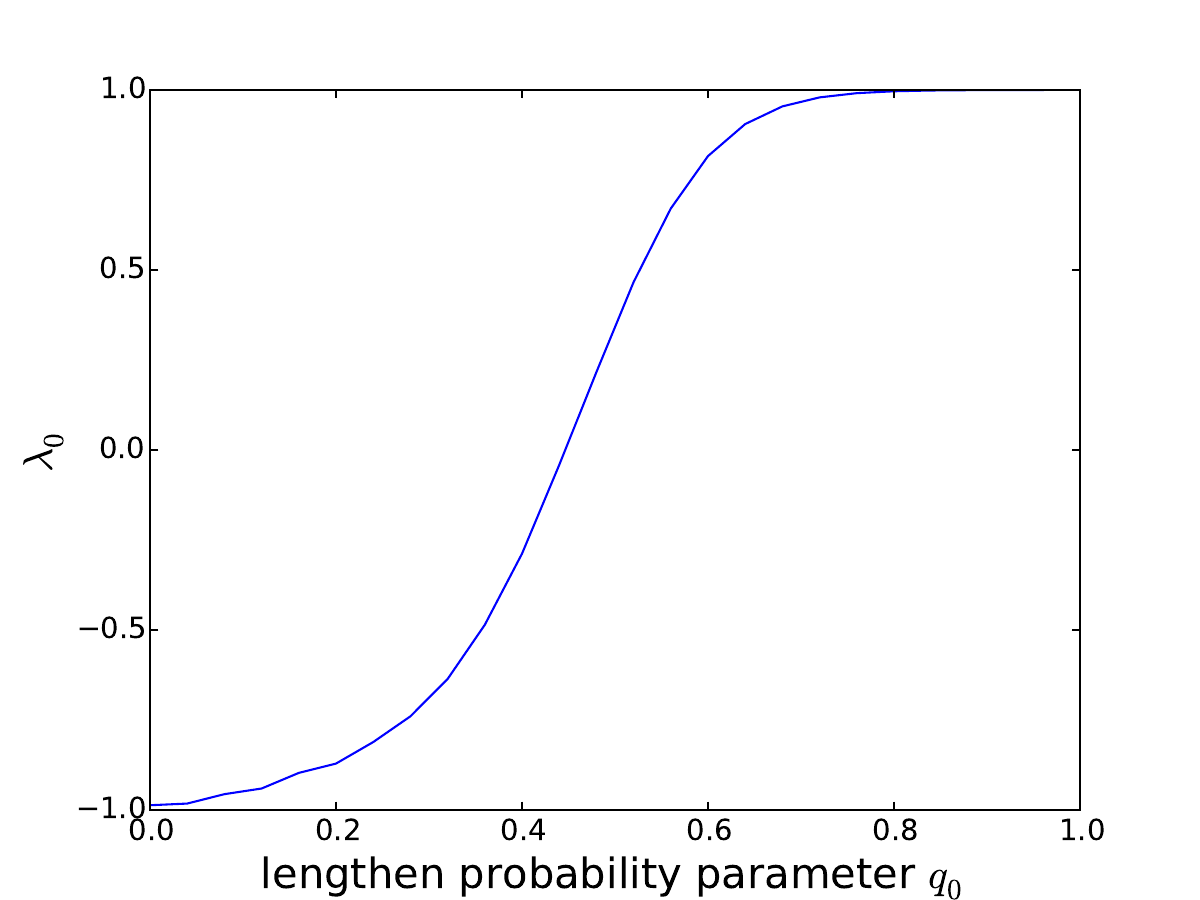}
        \end{subfigure}
        \begin{subfigure}{0.33\textwidth}
            \includegraphics[width=5.5cm]{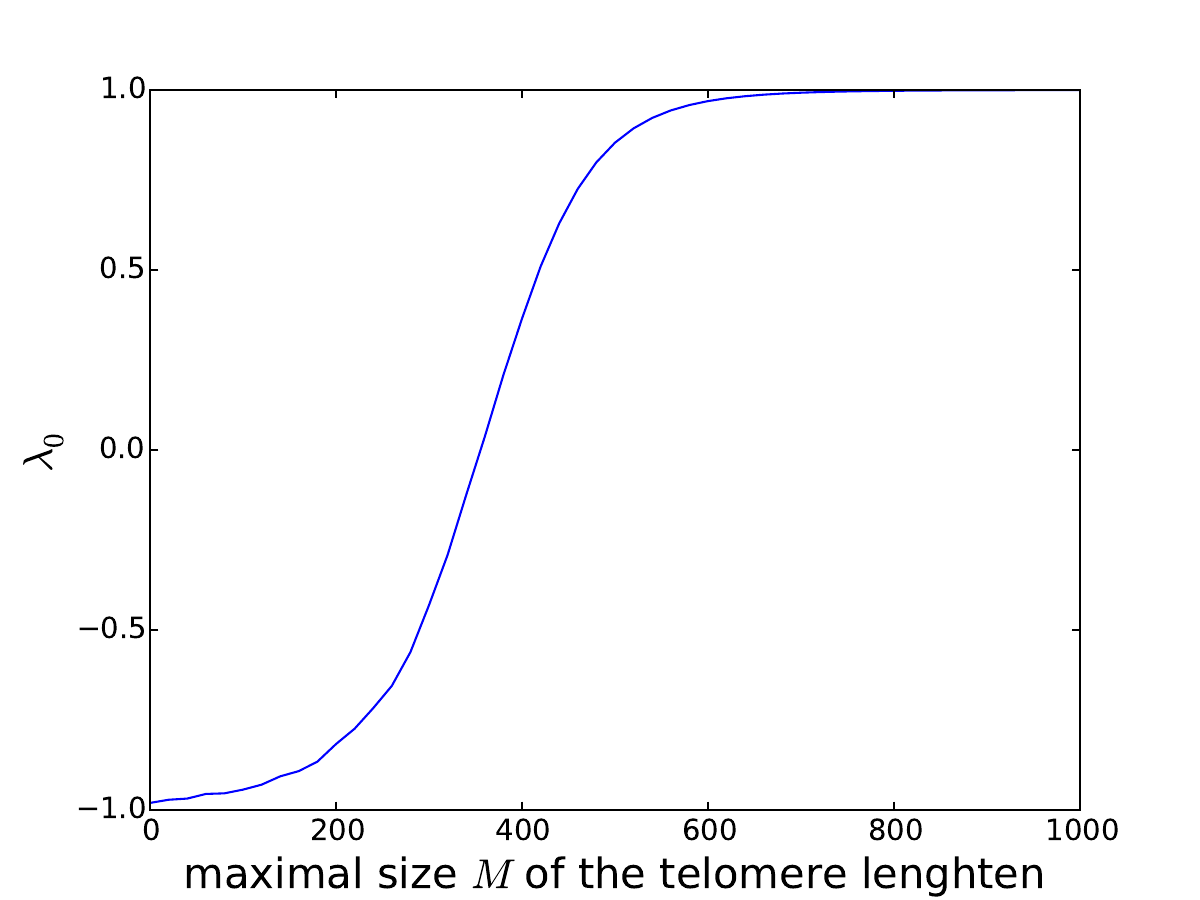}
        \end{subfigure}
        \begin{subfigure}{0.32\textwidth}
            \includegraphics[width=5.5cm]{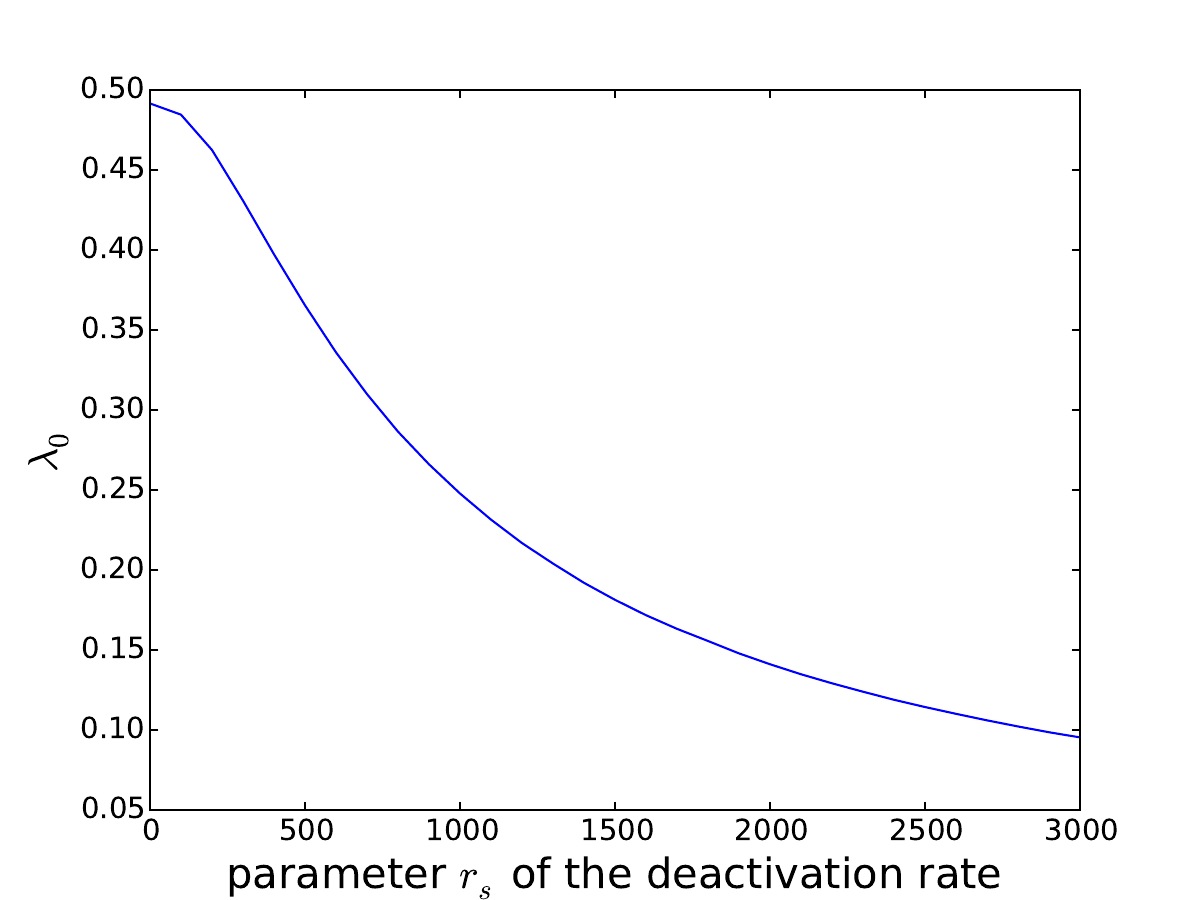}
        \end{subfigure}
    \caption{\label{fig:Lambda0_wrt_para}  Malthusian parameter $\lambda_0$~\eqref{eq.def.lambda0} for different values of $q_0$ (left), $M$ (center) and $r_s$ (right),  for the model with the telomerase mechanism given by \eqref{eq:param1telo} with the default values $q_0=0.7$, $M=600$ and all other parameters being as in~\eqref{eq:param1} and~\eqref{eq:param2}. }
\end{figure}

As previously mentioned, we also consider parameter values which yield $\lambda_0 < 0$, i.e. a subcritical population. In particular, in Figure~\ref{fig:Lambda0_wrt_para}, we see that certain values of $q_0$ and $M$ yield a subcritical population, despite the presence of telomerase. In this case, the asymptotic average population size of active cells decays exponentially until extinction, so that the total population size eventually plateaus out, and thus the Hayflick limit is finite.

To study this in more detail, in Figure~\ref{fig:HayFctq0M}, we represent the Hayflick limit as a function of the parameters $q_0$ and $M$ for these values.  As expected, the monotonic behaviour is the same as in Figure~\ref{fig:Lambda0_wrt_para}. 

\begin{figure}
        \begin{subfigure}{0.49\textwidth}
            \includegraphics[width=6.9cm]{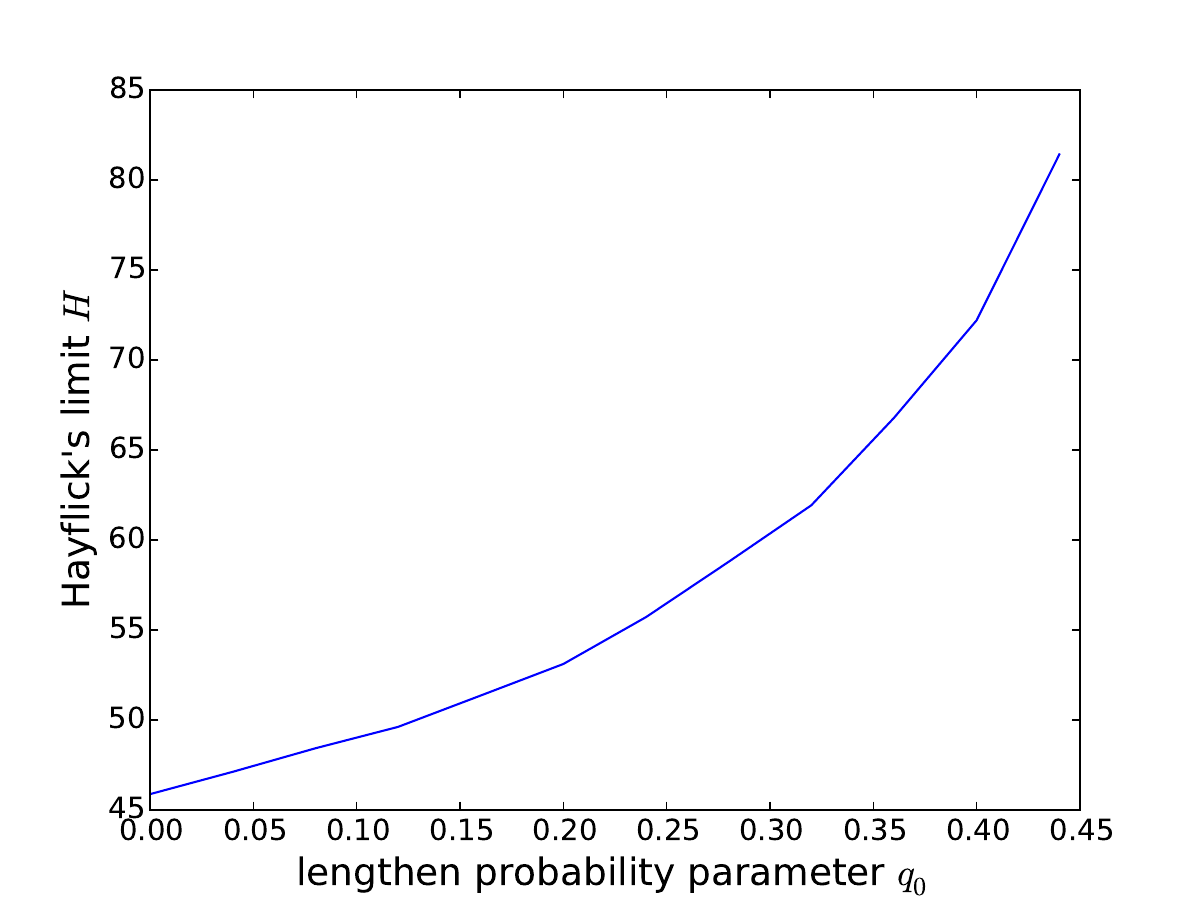}
        \end{subfigure}
        \begin{subfigure}{0.49\textwidth}
            \includegraphics[width=6.9cm]{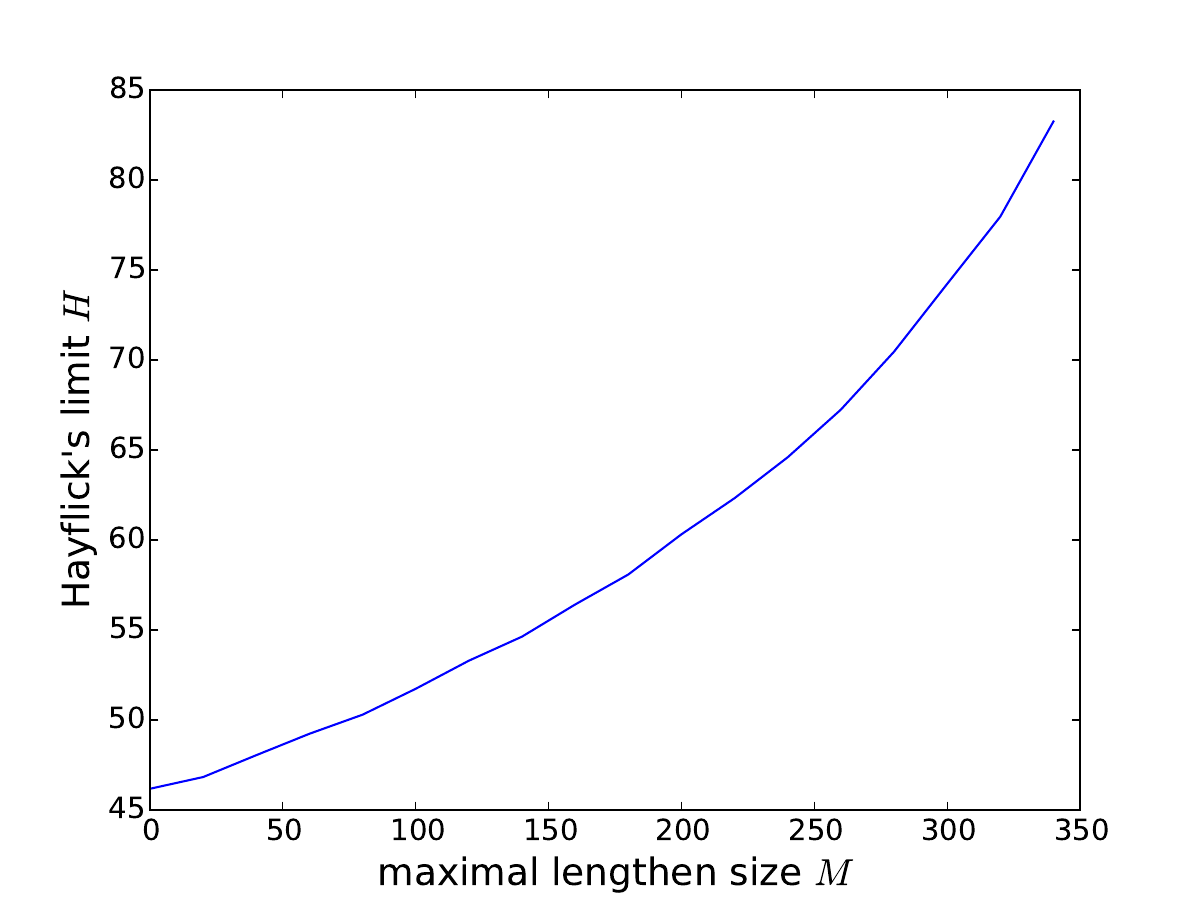}
        \end{subfigure}
    \caption{\label{fig:HayFctq0M}  Hayflick limits~\eqref{eq:hayflick} for different values of $q_0$ (left) and $M$ (right) leading to $\lambda_0 < 0$,  for the model with the telomerase mechanism given by \eqref{eq:param1telo} with the default values $M=600$ (left picture) and $q_0=0.7$ (right picture) and all other parameters being as in~\eqref{eq:param1} and~\eqref{eq:param2}. }
\end{figure}

On the other hand, when $\lambda_0$ is positive, the Hayflick limit is infinite. This is the case, for example, for the default set of parameters given at the start of the subsection. Figure~\ref{fig:pop_dyn_K46} represents the evolution of certain functions of the distribution of telomere lengths in cells for this default set of parameters. For convenience, we have plotted the distribution of the minimum, maximum and mean (in each cell) of the telomere lengths.
We observe the convergence of the distribution toward the distribution $\nu$, as stated in Theorem~\ref{thm:mainresult}.

Finally, in Figure~\ref{fig.QSD.wrt.K}, we have plotted these distributions at the final time of the simulation for different values of $K$. 
We see that, the higher the number of chromosomes $K$, the higher the variance of the telomere length distribution.  In particular, this figure shows that as the number of chromosomes increases, the average of the distribution of the minimum telomere length decreases. Since in addition, the deactivation and removal of cells from the system are determined by the length of the shortest telomeres in the cells, this entails that $\lambda_0$ decreases with $K$, as demonstrated by  Table~\ref{table:lambda_0_wrt_K}.

\begin{figure}
\includegraphics[width=12cm, trim = 4.4cm 0cm 1.9cm 0cm, clip=true]{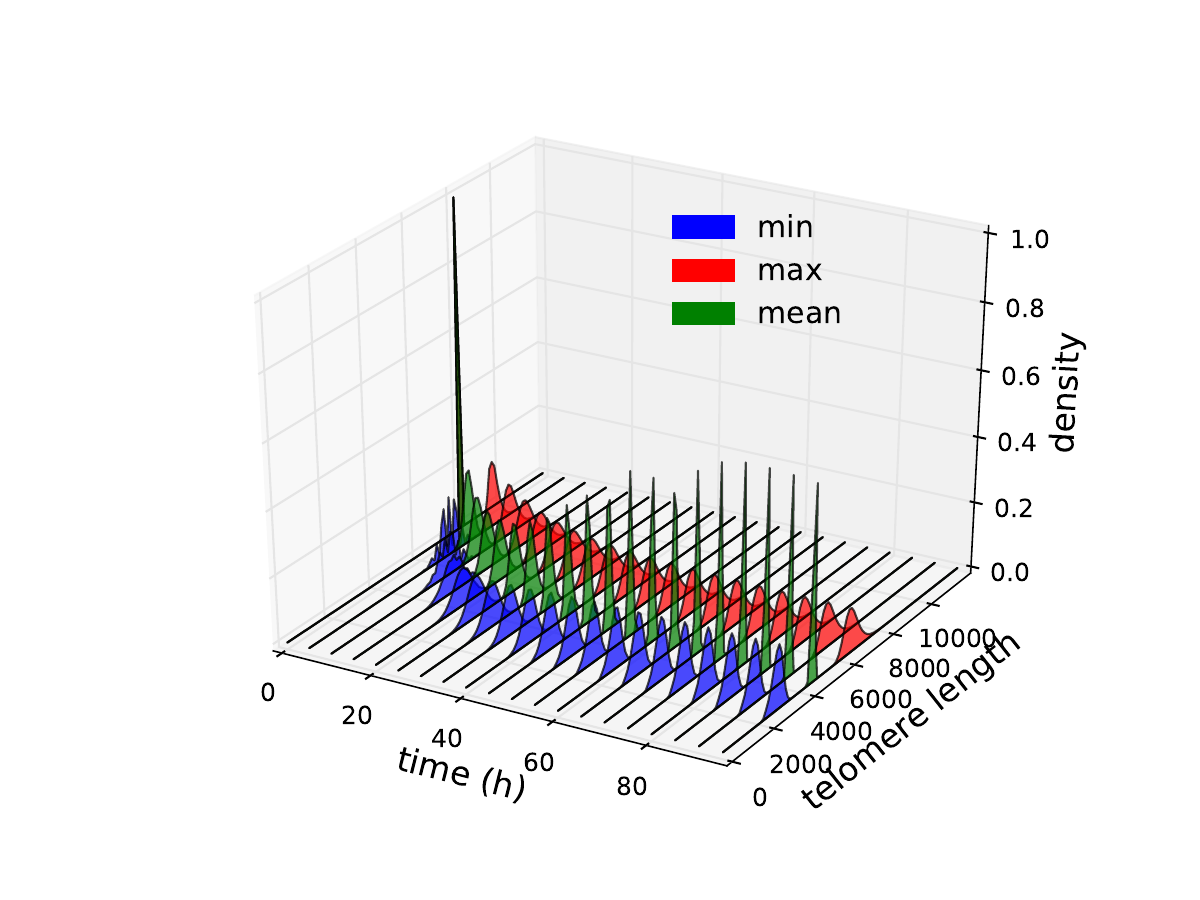}
\caption{\label{fig:pop_dyn_K46}Dynamics of certain characteristics of the (normalized) population density for the model with the telomerase mechanism given by \eqref{eq:param1telo} with the default values $q_0=0.7$ and $M=600$ and all other parameters being as in~\eqref{eq:param1} and~\eqref{eq:param2} (in particular with $K=46$). More precisely, we plot the normalization of the distributions $\mathbb E_{(c_0,\nsen)}(\sum_{i=1}^{N_t}\mathbf 1_{g(c_i(t), x_i(t))\in\cdot})$, with $g(c_i(t),x_i(t))=\mathbf 1_{x_i(t)=\nsen}\min c_i(t)$ (blue), $g(c_i(t),x_i(t))=\mathbf 1_{x_i(t)=\nsen}\max c_i(t)$ (red) and $g(c_i(t),x_i(t))=\mathbf 1_{x_i(t)=\nsen}\text{mean } c_i(t)$ (green), where $\text{mean } c_i(t)$ denotes the average telomere length in $c_i(t)\in (\mathbb N\times \mathbb N)^K$.}
\end{figure}

\begin{figure}
	\begin{subfigure}{0.215\textwidth}
 	\includegraphics[height=3.3cm]{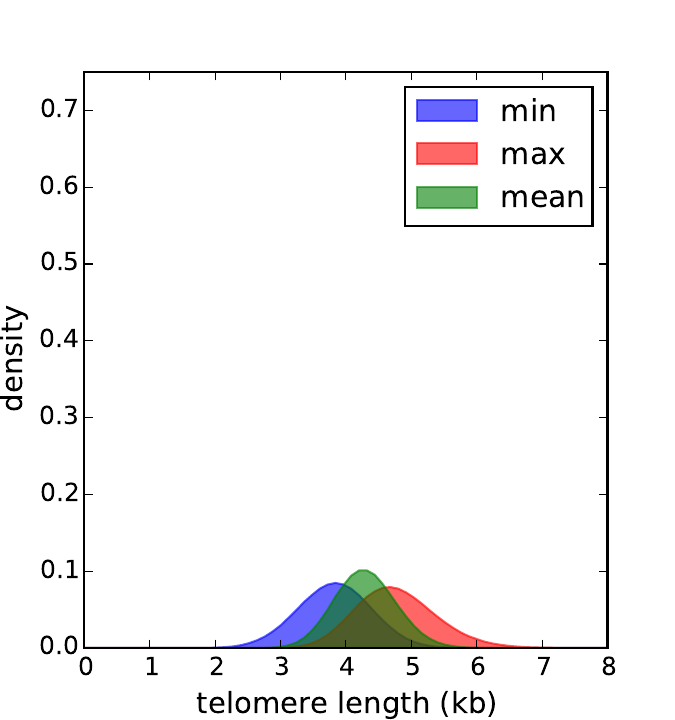}
	\caption{$K=1$.}
	\end{subfigure}
	\begin{subfigure}{0.188\textwidth}
 	\includegraphics[height=3.3cm, trim = 1.35cm 0cm 0cm 0cm, clip=true]{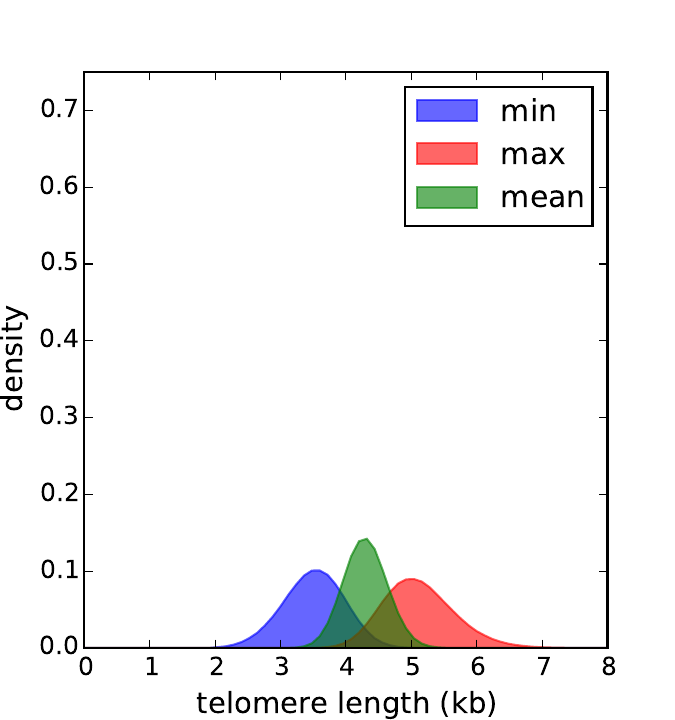}
	\caption{$K=2$.}
	\end{subfigure}
	\begin{subfigure}{0.188\textwidth}
 	\includegraphics[height=3.3cm, trim = 1.35cm 0cm 0cm 0cm, clip=true]{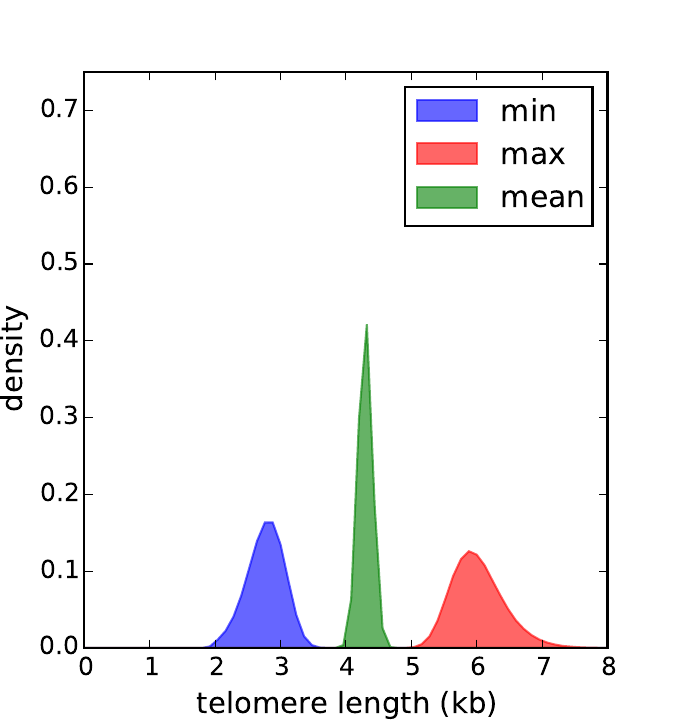}
	\caption{$K=20$.}
	\end{subfigure}
	\begin{subfigure}{0.188\textwidth}
 	\includegraphics[height=3.3cm, trim = 1.35cm 0cm 0cm 0cm, clip=true]{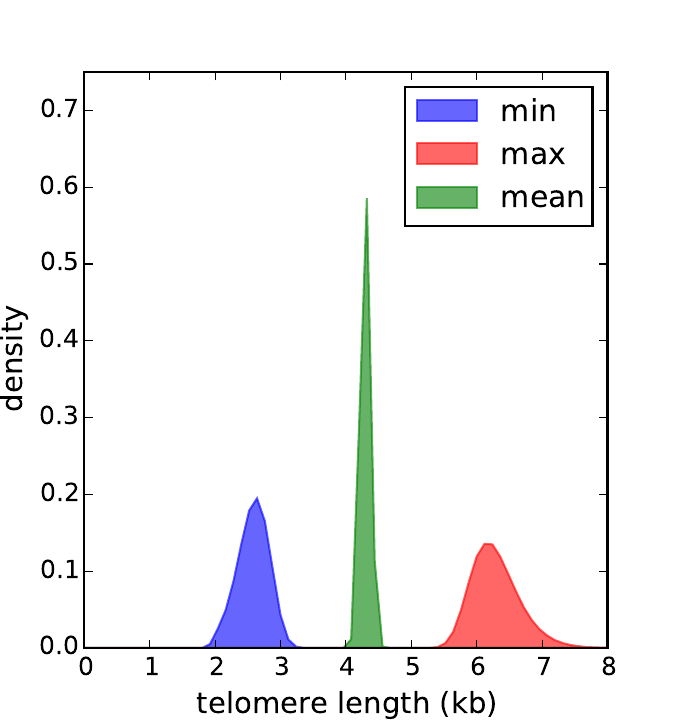}
	\caption{$K=46$.}
	\end{subfigure}
	\begin{subfigure}{0.188\textwidth}
 	\includegraphics[height=3.3cm, trim = 1.35cm 0cm 0cm 0cm, clip=true]{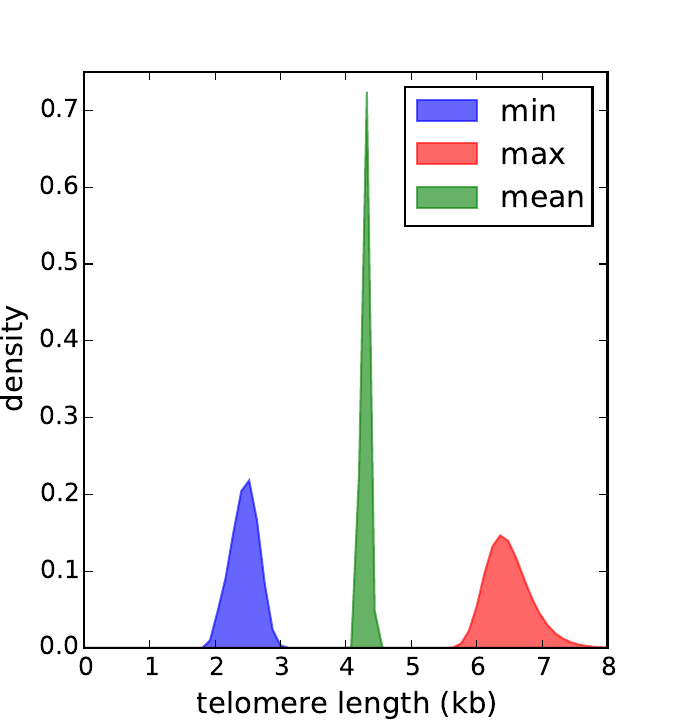}
	\caption{$K=90$.}
	\end{subfigure}
\caption{\label{fig.QSD.wrt.K}  Distribution of the minimal (blue), mean (green) and maximal (red) lengths of telomere for the limit distribution $\nu$ given by Theorem~\ref{thm:mainresult}, for different values of the number of chromosomes $K$ in cells for the model with the telomerase mechanism given by \eqref{eq:param1telo} with the default values $q_0=0.7$ and $M=600$ and all other parameters being as in~\eqref{eq:param1} and~\eqref{eq:param2}. }
\end{figure}

\begin{table}[h!]
\begin{tabular}{c|c}
Chromosomes number $K$ & Malthusian parameter $\lambda_0$\\
\hline
1 & 0.9992918\\
2 & 0.9986076\\
20 & 0.9863393\\
46 & 0.9692310\\
90 & 0.9416240 
\end{tabular}
\caption{\label{table:lambda_0_wrt_K}Malthusian parameter $\lambda_0$ defined in~\eqref{eq.def.lambda0} for the different values of $K$ used in Figure~\ref{fig.QSD.wrt.K}.}
\end{table}

\section{Discussion}
\label{sec:discussion}
We built a probabilistic individual based model for the telomere dynamics in a population of cells. We studied theoretically and numerically the dynamics of this model, with and without  telomerase. In particular, we estimated numerically the effect of several parameters on the Hayflick limit and the growth of the size of the cell population. 

Although the probabilistic model was built to mimic the biological mechanisms involved in telomere length dynamics, we made several assumptions in order to make the model more tractable for theoretical and numerical analysis, and in particular, to keep the number of parameters in the model sufficiently small. However, in some situations, one may wish to consider the following mechanisms: abrupt telomere shortening \citep{vspoljaric2019mathematical}, 
single strand breaks \citep{von2000accumulation} and oxidative  stress \citep{von2002oxidative, arkus2005mathematical}. This would require one to modify our model to allow $\overhang$ to be random and to include rare abrupt events. 
In our model, we assumed that chromatids are divided randomly between daughter cells during mitosis and there is no biological evidence against it. In some organism stem cells, biased DNA segregation during mitosis has been observed in vitro \citep{karpowicz2009germline, ranjan2022}. These observations only concern stem cells and still need to be confirmed in vivo, but may also be included in our model by removing the assumption that the random variables $B_j$ are not independent.

Additional mechanisms may also be studied mathematically and numerically. Namely, one may take into account the reactivation of telomerase in inactive cells (e.g. Telomerase reverse transcriptase in cancer cells~\citep{Dratwa2020}), fluctuation of the telomerase intensity through time or the population (depending on the environment or on the state of the population, leading to probabilistic models in random environments or with density dependence) or on/off  commutator mechanisms, typically modelled by a Piecewise Deterministic Markov Process representing the density of a commutator enzyme (see~\cite{tucey2014regulated} for a biological description of commutators). To model logistic constraints, one may also wish to consider a spatial model (see e.g.~\cite{David2021}) or include the dynamics of a shared limited resource in the model (as e.g. for chemostat models, see~\cite{fritsch2015modeling}). Finally, one may study the influence of the dependence of the overhang $\omega$ on the length of the telomere, empirically observed by~\cite{huffman2000telomere,rahman2008telomeric}. 

We finally mention the possibility of modelling an age-dependent process. At the expense of additional technicalities and an increase in model complexity, one may choose to study an age-dependent model, where the different rates described in our model depend on the age of the cell, see e.g.  the work of \cite{Olaye2024}. From a probabilistic perspective, this lies in the setting of Crump-Mode-Jagers branching processes~\citep{jagers1975branching}.

In our numerical simulations, we used parameter values inspired by empirical measures from the biological literature. Another natural approach, which will be the subject of further research, is to estimate the parameters from empirical data.
Since most experimental data concerning human cells are available for differentiated cells, such as leukocytes (see e.g.~\cite{Toupance2019} where the distribution of telomere length in humans is studied), in future work, we will study  population processes to model multi-tissue organisms where stem cells differentiate into specialized cells.

\section*{Statements and Declarations}

\subsection*{Acknowledgements}
The authors would like to thank Anne G\'egout-Petit for her support and useful discussions throughout the project.
This work was supported partly by the french PIA project ``Lorraine Universit\'e d'Excellence'', ANR-15-IDEX-04-LUE and received government funding managed by the Agence Nationale de la Recherche under the France 2030 program, reference ANR-23-EXMA-0005. 
Simulations are run on the babycluster of the Institut \'Elie Cartan de Lorraine.

\subsection*{Competing interest}
The authors have no relevant financial or non-financial interests to disclose.

\subsection*{Compliance with Ethical Standards}
The authors report there are no conflict of interest to declare.

\subsection*{Data availability statement}
No datasets were generated or analysed during the current study.

\appendix
\section{Dynkin's lemma}\label{sec:Dynkin}
In this section, we discuss a useful result that is used in several places in the article. The full version of this result and further details can be found in \cite[Theorem 2.1]{NTEbook}. We note that this result is sometimes referred to as Dynkin's lemma due to the fact that a slightly less general version of this result was first introduced by \cite{Dynkin}.   

Let us first introduce some necessary notation. Let $E$ be a Polish space (it is possible to make $E$ even more general but it will not be necessary for this article), which will denote the state space. We will also append an extra point $\dagger \notin E$ to $E$, that will denote the ``cemetary'' state, i.e. where particles are sent when they die. We let $\xi = (\xi_t, t \ge 0)$ denote a Markov process on $E$ with lifetime $\zeta := \inf\{t \ge 0 : \xi_t = \dagger\}$. Note that we consider $\dagger$ to be an absorbing state meaning that $\xi_t \in \{\dagger\}$ for all $t \ge \zeta$. We also let $\mathbf P_x$ denote the law of $\xi$ conditional on $\xi_0 = x$ with corresponding expectation operator $\mathbf E_x$.

In what follows, we will let $B(E)$ denote the space of bounded, measurable functions on $E$ such that $f(\dagger) = 0$ for all $f \in B(E)$. We will also let $B^+(E)$ denote the subset of $B(E)$ consisting only of non-negative functions. Then, for $f \in B^+(E)$, $t \ge 0$ and $x \in E$, we define
\[
  \mathtt Q_t[f](x) := \mathbf E_x[f(\xi_t)]
\]
to be the expectation semigroup associated to $(\xi, \mathbf P_x)$. We now extend the notion of expectation semigroups. Suppose that $\gamma \in B(E)$ and define, for $x \in E$, $t \ge 0$ and $f \in B(E)$, 
\[
  \mathtt Q_t^\gamma[f](x) := \mathbf E_x\left[{\rm e}^{\int_0^t \gamma(\xi_s){\rm d}s}f(\xi_t) \right].
\]
The weight ${\rm e}^{\int_0^t \gamma(\xi_s){\rm d}s}$ is often referred to as a multiplicative potential. This potential can be thought of as weighting the Markov process, consequently giving path a notion of importance. For example, if $\sup_{x \in E}\gamma(x) \le 0$, then the multiplicative potential acts as a penalisation term, since this weight will be less than or equal to one. 

Often, the expectation semigroup $(\mathtt Q_t^\gamma, t \ge 0)$ forms the basis of an evolution equation:
\begin{equation}\label{eq:Dynkin1}
    \chi_t(x) := \mathtt Q_t^\gamma[g](x) + \int_0^t \mathtt Q_s^\gamma[h_{t-s}](x) {\rm d}s, \qquad t \ge 0, x \in E, g \in B^+(E), 
\end{equation}
and where $h: [0, \infty) \times E \to [0, \infty)$ satisfies $\sup_{s \le t}|h_s| \in B^+(E)$. 

We are now in a position to state a result that will be useful throughout the article, which offers an alternative form of \eqref{eq:Dynkin1}.

\smallskip

\begin{thm}\label{thm:Dynkin}
Suppose that $|\gamma| \in B^+(E)$ and that $\sup_{s \le t}|h_s| \in B^+(E)$ for all $t \ge 0$. If $(\chi_t, t \ge 0)$ is represented by \eqref{eq:Dynkin1}, then it also solves
\begin{equation}\label{eq:Dynkin2}
  \chi_t(x) = \mathtt Q_t[g](x) + \int_0^t \mathtt Q_s[h_{t-s} + \gamma \chi_{t-s}](x){\rm d}s, \qquad t \ge 0, x \in E.  
\end{equation}
\end{thm}

\smallskip

Note that under an extra condition, the converse of this statement is true however, we will not state the details since we only use the result given above in this article. 

The full proof of the above result is essentially integration by parts and we refer the interested reader to \cite{NTEbook}. A quick calculation that can help the reader to convince themselves of this result is to compute the generator associated with \eqref{eq:Dynkin1} to obtain a differential version of \eqref{eq:Dynkin2}.

The above result can be thought of as a form of calculus by which the multiplicative potential in $\mathtt Q^\gamma$ is removed and appears instead as an additive potential in the integral term. Heuristically, it shows us how to transfer mass within the evolution equation, resulting in two different stochastic interpretations of the solution to \eqref{eq:Dynkin1} or equivalently, \eqref{eq:Dynkin2}. For simplicity, let us suppose that $h \equiv 0$. Then, on the one hand, \eqref{eq:Dynkin2} suggests a solution given by the expectation of a branching process whose particles move according to the semigroup $\mathtt Q$ and mass is created according to $\gamma$. On the other hand, \eqref{eq:Dynkin1} offers another interpretation of the solution as the expectation of a single particle, weighted by the potential $\gamma$, that is, the solution is given by the semigroup $\mathtt Q^\gamma$. The inclusion of $h_{t-s}$ in the integral term simply allows for a more general setting, which we now illustrate with an example. Further examples can be found in \cite[Chapter 2]{NTEbook}.

\medskip

\begin{example}
Consider the case where $\xi$ is a Brownian motion in $E = \mathbb R$ such that, at rate $\alpha \in B^+(\mathbb R)$, the particle jumps to a new position in $\mathbb R$ according to the law $\mu$. Hence, if $T$ is the time of the first such jump, then 
\begin{equation}\label{eq:rate}
    {\rm Pr}(T > t | \sigma(B_s, s \le t)) = {\rm e}^{-\int_0^t \alpha(B_s){\rm d}s}, \qquad t \ge 0.
\end{equation}
Let $(X_t, t \ge 0)$ denote this process with law $\mathbb P_x$ conditional on $X_0 = x \in \mathbb R$ and set $\chi_t(x) = \mathbb E_x[f(X_t)]$ for $f \in B^+(\mathbb R)$. By conditioning on whether $T \le t$ or $T > t$, it can be shown (see p.23 of \cite{NTEbook} for the details) that $\chi_t(x)$ satisfies \eqref{eq:Dynkin1} with 
\[
  \mathtt Q_t^{-\alpha}[f](x) = \mathbf E_x\left[{\rm e}^{-\int_0^t \alpha(B_s){\rm d}s}f(B_t) \right]
\]
and 
\[
  h_{t-s} = \alpha(\cdot)\int_{\mathbb R}\chi_{t-s}(y) \mu({\rm d}y).
\]
Theorem \ref{thm:Dynkin} then shows that, equally, $\chi_t(x)$ satisfies
\begin{equation*}
    \chi_t(x) = \mathbf E_x[f(B_t)] + \int_0^t \mathbf E_x\left[\alpha(B_s)\left(\int_{\mathbb R}\chi_{t-s}(y) \mu({\rm d}y) - \chi_{t-s}(x)\right)\right]{\rm d}s.
\end{equation*}
This second equation can be interpreted as follows. At time $t$, either no jump has occurred and so the particle is at position $B_t$, or with rate $\alpha$, the `parent' particle is replaced by a `new' particle whose position is given by $\mu$.
\end{example}

\section{Algorithm}\label{sec:algo}
Here we present the interacting particle approximation scheme used to produce the simulations in Section \ref{sec:simulations}. The following algorithm allows one to simulate a population with fixed size, $N$ say. Initially, each particle evolves according to an independent copy of a (sub)Markov process. When a particle is killed, it is resampled from the surviving population and the particles then continue to evolve independently. In this case, the process evolves according to $(\bar{\mathcal C}_t, \bar{\mathcal X}_t)_{t \ge 0}$ defined in the proof of Theorem \ref{thm:mainresult}. We refer the reader to the article of \cite{DelMoral2004, DelMoral2013} for further details of algorithms of this type.

\begin{rem}
\label{rem:generalization}
    The estimation of $\lambda_0$ from Algorithm~\ref{algo.ibm} by the methods of~\cite{DelMoral2004, DelMoral2013} is valid under the convergence stated in Theorem~\ref{thm:mainresult}.
    Despite the fact that Theorem \ref{thm:mainresult} was proved under restrictive assumptions (including $\lambda_0 > 0$), numerical simulations suggest that the convergence also holds true for all the parameters choices of Section~\ref{sec:simulations}. Based on these considerations, we allow ourselves to use the same numerical methods to estimate the parameter $\lambda_0$ across Section~\ref{sec:simulations}. In particular, we
    consider certain parameter regimes where $\lambda_0 \le 0$, in order to illustrate that our result should hold under less restrictive assumptions, and to provide the reader with a more complete picture.  We leave the question of the possible generalization of Theorem~\ref{thm:mainresult} to these parameters open.
\end{rem}

\begin{algorithm}[H]
\small
\begin{algorithmic}
\STATE sample $(c_i(0),x_i(0))_{i \in\llbracket 1,N\rrbracket}=\left(c_0, \nsen\right)_{i\in\llbracket 1,N\rrbracket}$
 \COMMENT{initialization}
\STATE $t\leftarrow 0$
\WHILE {$t\leq T$}
  \STATE $\tau \leftarrow \sum_{i=1}^N 1_{x_i(t)=\nsen} (2b_{\nsen}(c_i(t)) + s_{\nsen}(c_i(t)) + d_{\nsen} (c_i(t))+
  					1_{x_i(t)=\sen} d_{\sen} (c_i(t))$
  \STATE $\Delta t \sim Exp(\tau)$  
  \STATE $i \sim \sum_i \frac{1_{x_i(t)=\nsen} (2b_{\nsen}(c_i(t)) + s_{\nsen}(c_i(t)) + d_{\nsen} (c_i(t))+
  					1_{x_i(t)=\sen} d_{\sen} (c_i(t))}{\tau} \delta_i$
  
  \STATE $(c_j(t+\Delta t), x_j(t+\Delta t))_{j\in\llbracket 1,N\rrbracket, j\neq i}=(c_j(t), x_j(t))_{j\in\llbracket 1,N\rrbracket, j\neq i}$
  
  \STATE $u\sim \mathcal{U}[0,1]$ 
  \IF {$x_i(t)=\nsen$ and $u\leq \frac{2b_{\nsen}(c_i(t))}{2b_{\nsen}(c_i(t)) + s_{\nsen}(c_i(t)) + d_{\nsen} (c_i(t))}$}
      \STATE \COMMENT{cell division} 
      \STATE $(m_{j},n_{j})_{j\in\llbracket 1,K\rrbracket} = c_i(t)$
      \STATE $B_j \sim Ber(\nicefrac 12)$, $j=1, \dots, K$
	 \STATE $L_j \sim \mu_f$, $L'_j \sim \mu_f$, $j=1, \dots, K$
      \STATE $\chi_j \sim Ber(q(m_{j}-\overhang B_j))$, $j=1, \dots, K$
      \STATE $\chi'_j \sim Ber(q(n_{j}-\overhang (1-B_j)))$, $j=1, \dots, K$
      \STATE $m_{j} \leftarrow m_{j}-\overhang B_j + L_j \chi_j$, $j=1, \dots, K$
      \STATE $n_{j} \leftarrow n_{j}-\overhang (1-B_j) + L'_j \chi'_j$, $j=1, \dots, K$
      \IF{ $\min_{1\leq j \leq K}\{m_{j},n_{j}\}\leq L_{\min}$}
      	 \STATE \COMMENT{non-viable cell}
      		\STATE $(c_i(t+\Delta t),x_i(t+\Delta t)) \leftarrow ((m_{j},n_{j})_{j\in\llbracket 1,K\rrbracket }, \sen)$
      \ELSE
      	  \STATE \COMMENT{viable cell}
	      \STATE $(c_i(t+\Delta t),x_i(t+\Delta t)) \leftarrow ((m_{j},n_{j})_{j\in\llbracket 1,K\rrbracket }, \nsen)$
	   \ENDIF
  \ELSE\IF{$x_i(t)=\nsen$ and $u\leq \frac{2b_{\nsen}(c_i(t))+s_{\nsen}(c_i(t))}{2b_{\nsen}(c_i(t)) + s_{\nsen}(c_i(t)) + d_{\nsen} (c_i(t))}$}
  		\STATE  \COMMENT{cell deactivation}
  		\STATE $(c_i(t+\Delta t),x_i(t+\Delta t))\leftarrow (c_i(t),\sen)$
  \ELSE
  		 \STATE \COMMENT{cell removal: the cell is uniformly sampled on existing cells}
      \STATE $(c_i(t+\Delta t),x_i(t+\Delta t)) \sim \mathcal{U}((c_j(t+\Delta t),x_j(t+\Delta t))_{j\in\llbracket 1,K\rrbracket, j\neq i})$
  \ENDIF\ENDIF
  \STATE $t \leftarrow t+\Delta t$
\ENDWHILE
\end{algorithmic}
\caption{\label{algo.ibm} Fleming-Viot type scheme: simulation of $N$ processes $(\mathcal{C}_t, \mathcal{X}_t)_{t \ge 0}$ defined in Section~\ref{sec:manytoone}. Each process is simulated by a Gillespie Algorithm \citep{gillespie1977a}. When one process goes extinct, it is replaced by one alive cell, drawn uniformly among the alive processes.}
\end{algorithm}

\section{Notation}
\label{sec:notation}
Here we provide the reader with a table of notation for the parameters that are used in the numerical simulations in Section~\ref{sec:simulations}.
 We give the notation used, a description and the page number where it was first introduced.
\begin{table}[h]
\begin{tabular}{c l l}
{\bf Notation} &  \multicolumn{1}{c}{\textbf{Description}}  &  {\bf Reference}\\
\hline 
$(m, n)$ & Chromosome indexed by its telomeres& \ p.\pageref{notation:mn}\\
$\overhang$ & Overhang & \ p.\pageref{notation:overgang}\\
$K$ & Number of chromosomes per cell &\ p.\pageref{notation:K}\\
$L_{min}$& Minimal telomere length of active cells & \ p.\pageref{notation:Lmin} \\
$(c,\nsen) = \left((m_{j}, n_{j})_{j = 1}^K,\nsen\right)$& 
Active cell & \ p.\pageref{notation:sennsen}\\
$(c,\sen) = \left((m_{j}, n_{j})_{j = 1}^K,\sen\right)$& 
Non active cell& \ p.\pageref{notation:sennsen}\\
$c_0$ & Telomere lengths of each initial cell & \ p.\pageref{notation:c0}\\
$q(\ell)$ & Probability that telomerase acts & \ p.\pageref{notation:qell}\\
$q_0$ & Parameter of $q(\ell)$ & \ p.\pageref{eq:param1telo} \\
$\mu_f$ & Telomere lengthen law & \ p.\pageref{notation:mu_f} \\
$M$ & Maximal telomere lengthen & \ p.\pageref{eq:param1telo} \\
$b_\nsen$ & Division rate & \ p.\pageref{notation:bnsen}\\
$s_\nsen$ & Deactivation rate & \ p.\pageref{notation:snsen}\\
$r_s$ & Parameter of $s_{\nsen} $ & \ p.\pageref{notation:rs}\\
$d_x$ & Removal rate & \ p.\pageref{notation:dx} \\
$N_t$& Cell number in the population & \ p.\pageref{notation:Nt}\\
$N_\infty$& Final population size & \ p.\pageref{notation:Ninfty}\\
$H$ & Hayflick limit & \ p.\pageref{notation:H} \\
$\lambda_0$ & Malthusian parameter & \ p.\pageref{eq.def.lambda0}\\
\end{tabular}
\end{table}

\end{document}